\documentclass[12pt]{elsarticle}
\usepackage[letterpaper,
    vmargin=1in,
    head=16pt]{geometry}
\usepackage[pdftex,
    colorlinks=false]{hyperref}
\usepackage{url}
\usepackage{array}
\usepackage{booktabs}
\usepackage{amsmath}
\usepackage{latexsym}
\usepackage{amsthm}
\usepackage{amssymb}
\usepackage{amscd}
\usepackage{microtype}
\makeatletter
\def\MT@register@subst@font{\MT@exp@one@n\MT@in@clist\font@name\MT@font@list
 \ifMT@inlist@\else\xdef\MT@font@list{\MT@font@list\font@name,}\fi}
\makeatother

\usepackage{bm}
\usepackage{mathrsfs}
\usepackage{graphicx}
\usepackage{color}
\usepackage{slashed}

\theoremstyle{plain}
\newtheorem{thm}{Theorem}[section]
\newtheorem{cor}[thm]{Corollary}
\newtheorem{lem}[thm]{Lemma}

\theoremstyle{definition}

\theoremstyle{remark}
\newtheorem{remark}[thm]{Remark}

\numberwithin{equation}{section}

\def\VEV#1{\left\langle #1\right\rangle}

\newcommand*{\thalf}{\tfrac{1}{2}}

\newcommand*{\cL}{\mathcal{L}}
\newcommand*{\cM}{\mathcal{M}}
\newcommand*{\Honehat}{\widehat{\Hone}}
\newcommand*{\Lpz}{L_{P_0}}

\newcommand*{\cMt}{\widetilde{\mathcal{M}}}

\newcommand*{\half}{\frac{1}{2}}
\newcommand*{\spin}{\mathcal{S}}
\newcommand*{\spinbasic}{\hat{\mathcal{S}}}
\newcommand*{\qbasic}{\hat{q}}
\newcommand*{\sympmetric}{\begin{pmatrix} 0 & I \\ -I & 0 \end{pmatrix}}
\newcommand*{\Nb}{\bar{N}}
\newcommand*{\gen}{\bm{1_{2}}}
\newcommand*{\cV}{\mathcal{V}}
\newcommand*{\sfA}{\mathsf{A}}

\newcommand*{\ba}{\mathbf{a}}
\newcommand*{\bb}{\mathbf{b}}
\newcommand*{\bbZ}{\mathbb{Z}}
\newcommand*{\bbC}{\mathbb{C}}
\newcommand*{\bbR}{\mathbb{R}}
\newcommand*{\metric}{\bm{g}}
\newcommand*{\dbar}{\bar{\partial}}

\newcommand*{\linebundle}{\mathscr{L}}
\newcommand*{\zbar}{\bar{z}}
\newcommand*{\hb}{\bar{h}}
\newcommand*{\Zsc}{Z_{\text{sc}}}
\newcommand*{\onep}{\mathbf{1}^{\perp}}

\newcommand*{\detz}{\det_{\onep}\Delta_{0}}
\newcommand*{\BEG}[1]{[BEG, #1]}
\newcommand*{\Hone}{H_1(\Sigma,\bbZ)}
\newcommand*{\Hoz}{H^{1,0}(\Sigma)}
\newcommand*{\Hzo}{H^{0,1}(\Sigma)}
\newcommand*{\Lho}{\Lambda^{1/2,1}(\Sigma)}	
\newcommand*{\Lhz}{\Lambda^{1/2,0}(\Sigma)}	
\newcommand*{\Lzh}{\Lambda^{0,1/2}(\Sigma)}	

\newcommand*{\Loz}{\Lambda^{1,0}(\Sigma)}

\newcommand*{\Lzz}{\Lambda^{0,0}(\Sigma)}
\newcommand*{\Ltil}{\tilde{\mathcal{L}}} 

\newcommand{\zanalog}{\mathfrak{z}_{\kappa}}
\newcommand{\zanalogp}{\mathfrak{z}_{\kappa'}}
\newcommand*{\Dtil}{\widetilde{D}}
 
\newcommand{\Azo}{A^{0,1}}

\newcommand{\Bzo}{B^{0,1}} 

\newcommand*{\nutil}{\tilde{\nu}}

\newcommand*{\Jacz}{J_{0}(\Sigma)}
\newcommand*{\Jacgo}{J_{g-1}(\Sigma)} 
\newcommand*{\cD}{\mathcal{D}}
\newcommand*{\cDt}{\widetilde{\mathcal{D}}}
\newcommand*{\thchar}[2]{\vartheta\left[\begin{smallmatrix} #1 \\ #2
\end{smallmatrix} \right]}

\DeclareMathOperator{\detp}{\det\nolimits^{\prime}}
\DeclareMathOperator{\Sym}{Sym}

\DeclareMathOperator{\DET}{DET} 
\DeclareMathOperator{\Diff}{Diff}

\DeclareMathOperator{\Index}{index} 

\DeclareMathOperator{\Met}{Met}

\DeclareMathOperator{\Sp}{Sp}

\DeclareMathOperator{\cclass}{c}
 
\DeclareMathOperator{\coker}{coker}

\DeclareMathOperator{\hol}{hol} 

\DeclareMathOperator{\pont}{p} 

\DeclareMathOperator{\real}{re} 
\DeclareMathOperator{\vol}{vol}
\DeclareMathOperator{\Teich}{Teich}
\DeclareMathOperator{\Teichspin}{Teich^{1/2}}
\DeclareMathOperator{\Teichodd}{Teich^{1/2}_{\text{odd}}}

\DeclareMathOperator{\MCG}{MCG}
\DeclareMathOperator{\Torelli}{Torelli}
\DeclareMathOperator{\Torellibig}{\widetilde{Torelli}}
\DeclareMathOperator{\Torelliodd}{Torelli^{1/2}_{\text{odd}}}
\DeclareMathOperator{\Mododd}{\mathcal{M}^{1/2}_{\text{odd}}}

\DeclareMathOperator{\Modspin}{\mathcal{M}^{1/2}}
\DeclareMathOperator{\Modonetwo}{\mathcal{M}_{1,2}}
\DeclareMathOperator{\Mod}{\mathcal{M}}
\DeclareMathOperator{\Id}{Id}

\DeclareMathOperator{\MSO}{MSO}
\DeclareMathOperator{\MU}{MU}
\DeclareMathOperator{\MSpin}{MSpin}
\DeclareMathOperator{\MString}{MString}

\begin{document}
\title{Beyond the String Genus\tnoteref{t1}} 

\author[oa]{Orlando Alvarez\corref{cor1}}
\ead{oalvarez@miami.edu}

\author[ims]{I.M.~Singer}
\ead{ims@math.mit.edu}

\address[oa]{Department  of Physics, University of Miami, P.O. Box 248046,
Coral Gables, FL 33124 USA}

\address[ims]{Department of Mathematics, Massachusetts Institute of Technology,
77 Massachusetts Avenue, Room 2-387, Cambridge, MA 02139 USA}

\tnotetext[t1]{The work of OA was supported by the National Science
    Foundation under grants PHY-0554821 and PHY-0854366.  The work of IMS
    was supported by two DARPA grants through the Air Force Office of
    Scientific Research (AFOSR): grant numbers FA9550-07-1-0555 and
    HR0011-10-1-0054.}
\cortext[cor1]{Corresponding author}

\begin{abstract}
    In an earlier work we used a path integral analysis to propose a
    higher genus generalization of the elliptic genus.  We found a
    cobordism invariant parametrized by Teichmuller space.  Here we
    simplify the formula and study the behavior of our invariant under
    the action of the mapping class group of the Riemann surface.  We
    find that our invariant is a modular function with multiplier just
    as in genus one.
\end{abstract}
\maketitle

\section{Introduction}
\label{sec:intro}

In reference \cite{Alvarez:2001zd} (BEG) we constructed a cobordism
invariant using a supersymmetric sigma model with target space $M$, a
string manifold, \emph{i.e.}, $\tfrac{1}{2} \pont_{1}(M)=0$.  The case
of genus one
\cite{Alvarez:1987de,Alvarez:1987wg,Pilch:1987en,Witten:1987bf} gives
the string genus\footnote{We use the string genus, \emph{a.k.a} the
Witten genus, instead of the elliptic genus used by topologists that
corresponds to a more complicated field theory.} \emph{\`a la}
\cite{Alvarez:1987de,Alvarez:1987wg}.  The cobordism invariant, the
semiclassical approximation $\Zsc(M)$ of the partition function, was a
function on odd spin Teichmuller space $\Teichodd(\Sigma)$.
Teichmuller space $\Teich(\Sigma) = \Met(\Sigma)/\Diff_{0}(\Sigma)$ is
the space of metrics of constant curvature\footnote{The sectional
curvature $k$ is related to the Riemann tensor by $R_{abcd}^{\Sigma} =
k(g_{ac}g_{bd} - g_{ad}g_{bc})$.  The Ricci scalar is given by
$R^{\Sigma} =2k$.  By curvature $-1$ we mean $k=-1$.} $-1$ on a
Riemann surface $\Sigma$ (genus $g>1$) divided by the connected
diffeomorphisms.  Odd spin Teichmuller space $\Teichodd(\Sigma)$ is
the covering space of odd spin structures over the space of metrics,
all divided by $\Diff_{0}(\Sigma)$, a normal subgroup of
$\Diff(\Sigma)$.

The original purpose of the present paper was to refine our previous
results so that our cobordism invariant would be a function over spin
moduli space as opposed to spin Teichmuller space.  We attempted do so
by dividing our previous result by the action of the mapping class
group $\MCG(\Sigma) = \Diff(\Sigma)/\Diff_{0}(\Sigma)$.

We remind the reader of the exact sequence,
\begin{equation}
   1 \to \Torelli(\Sigma) \to \MCG(\Sigma) \to \Sp(2g, \bbZ) \to 
   1,
   \label{eq:MCG-seq}
\end{equation}
where $\Torelli(\Sigma)$ is the normal subgroup of the mapping class
group that is constant on $H_{1}(\Sigma,\bbZ)$ and $\Sp(2g, \bbZ)$ is
the symplectic group.  This suggested that we first examine the action
of $\Torelli(\Sigma)$ on our invariant and afterwards the action of
the symplectic group on what remains. 

\bigskip\sffamily
One of our main results is that $\Zsc$ is a \emph{modular} function on 
$\Teichodd(\Sigma)$ with multiplier (Section~\ref{sec:results}).
\normalfont

\bigskip
We outline  the contents of the paper.

In Section~\ref{sec:torelli} we analyze the effect of the Torelli
group on $\Zsc$; we are left with the effect of $\Sp(2g,\bbZ) =
\MCG(\Sigma)/\Torelli(\Sigma)$ on $\Zsc$. 

In Section~\ref{sec:RS-torsion} we generalize Ray-Singer torsion to
the spinor case.  We do so not only for its intrinsic interest but also
because it has \eqref{eq:det-univ} as a corollary.

In Section~\ref{sec:bosonic-det} we focus on equation \eqref{eq:def-Zsc-1} 
shown below
\begin{equation*}
    \Zsc = \left( \frac{\vol \Sigma}{\detz} \right)^{n}
    \left(\frac{\detp\partial_{\delta}}{\Nb_{\delta}^{2}} \right)^{n}
    \int_{M} \prod_{k=1}^{n} \frac{\zanalog\left(x_{k}
    z(\hb_{\delta}^{2})\right)}{\vartheta[\kappa] \left(x_{k}
    z(\hb_{\delta}^{2})\right)}.
\end{equation*}
In the above $\detz$ is the determinant of the laplacian on the space
of functions of the surface $\Sigma$ orthogonal to the constants.  An
odd spin structure $\delta$ is a square root $\bar{K}^{1/2}_{\delta}$
of the anti-canonical line bundle $\bar{K}$ of $\Sigma$.  The operator
$\partial_{\delta}$ maps $\Lambda^{0,1/2}(\Sigma)$ into
$\Lambda^{1,1/2}(\Sigma)$ and is of index zero.  Note that
$\partial_{\delta}$ gives a family of elliptic operators parametrized
by $\Met(\Sigma)$ and ultimately parametrized by $\Teichodd(\Sigma)$.
It has generically a one dimensional kernel\footnote{When
$\dim\ker\partial_{\delta}>1$, $\detp\partial_{\delta}=0$ and
therefore $\Zsc=0$.  The places where this occurs is a subvariety of
spin Teichmuller space with complex codimension $1$.} generated by
$\hb_{\delta}$.  Following Quillen~\cite{Quillen:det-line-bundle}
$\detp\partial_{\delta}$ is a section of the hermitian determinant
line bundle of the $\partial_{\delta}$ family and $\lvert N_{\delta}
\rvert = \lVert \hb_{\delta} \rVert$.  Riemann surface theory gives us
an explicit expression for the square of the spinor:
\begin{equation*}
    h_{\delta}^{2} = \sum_{k}\frac{\partial
    \vartheta[\delta](0)}{\partial z^{k}} \omega_{k}\,.
\end{equation*}
$(\omega_{1}, \dotsc, \omega_{g})$ is a symplectic standard basis for
the abelian differentials and $\vartheta[\delta](\cdot)$ is the
Riemann theta function with characteristic $\delta$.

We now describe the integral in $\Zsc$.  $\kappa$ is the vector of
Riemann constants and $\vartheta[\kappa](0;\Omega)=0$.  The theta
divisor $\Theta_{\kappa}$ near the origin of the Jacobian $\Jacz$ is
the zero set of $\zanalog$, see
\ref{sec:app-kappa} for the details.  The $x_{k}$ are the
formal eigenvalues of the curvature $2$-form on $M$ whose symmetric
polynomials express the Pontryagin classes of $M$.  Thus the integral
is a cobordism invariant depending on the metric of the
Riemann surface.

In Section~\ref{sec:symplectic}, we review the properties of
$\Sp(2g,\bbZ)$ and compute the transformation properties of various
quantities that appear in $\Zsc$.  Using results in
\cite{Fay:theta,Mumford:1} we find the transformation laws for the
integral term in $\Zsc$ and for $h_{\delta}^{2}$.

In Section~\ref{sec:quarfs} we give an exposition of quadratic
refinements of cup product and its one-to-one correspondence with spin
structures \cite{Atiyah:spin}.  The action of the symplectic group on
spin structures is derived by knowing its action on quarfs. We also 
introduce the subgroup $\Gamma_{1,2} \subset \Sp(2g,\bbZ)$.

In Section~\ref{sec:spinor-det} we show that the symplectic action
lifts to an action on $\Teichodd(\Sigma)$, the covering space of odd
spin structures.  There is a preferred even spin structure
$\spinbasic$ and we find equation \eqref{eq:detp-section}
\begin{equation*}
    \frac{\det\nolimits' \dbar_{\delta}}{N_{\delta}^{2}} =
    \frac{1}{\pi}\; \frac{\det
    \dbar(\spinbasic)}{\vartheta(0)}\,,
\end{equation*}
where $\delta = \Omega a + b$ and $\Omega$ is the Riemann period
matrix.  
The methods of this section give a generalization of 
the bosonization theorem to odd spin structures.

In Section~\ref{sec:geom-symp-det} we study the determinant for the
Dirac laplacian and we note that $f \in \Diff(\Sigma)$ (representing
an element $\Lambda \in \Sp(2g,\bbZ)$) induces an isometry between two
$f$-related determinant line bundles.  As a consequence, a phase
factor $e^{i\xi(0,\Lambda)}$ appears in our computations.  When
$\Lambda \in \Gamma_{1,2}$ we can compute $e^{i\xi(0,\Lambda)}$.  This
section contains our main results, the simplest of which  is that
$\Zsc$ is a modular function with multiplier when $\Lambda \in
\Gamma_{1,2}$.

In \ref{sec:conventions} we discuss the conjugate linear isomorphism
between $\Lhz$ and $\Lho$.

In \ref{sec:app-kappa} we review some properties of the determinant 
line bundle for $\dbar$-operators and discuss $\zanalog$.

In \ref{sec:genus-1} we relate our abstract modular transformation 
result to explicit formulas in genus one.

Finally we have a small nomenclature of recurring symbols.

Section~\ref{sec:bosonic-det} and \ref{sec:app-kappa} has some overlap
with material in BEG. We rely heavily on the content of these sections
in this paper and therefore include it here along with an expanded
discussion of some topics in BEG.

\subsection{Some Questions and Speculations}
\label{sec:speculations}

We note that when the Riemann surface degenerates our invariant
factorizes.  How does $\Zsc$ fit with the modular geometry of Friedan
and Shenker \cite{Friedan:1986ua} that describes the behavior of
string amplitudes when the Riemann surface degenerates?

Our main results give a genus, a map from a subring of the string
cobordism ring to a subring of the functions on $\Teichodd(\Sigma)$.
For $g=1$ this leads to the string genus. Does our new genus give any 
new information about the string cobordism ring?

M.J.~Hopkins suggests that we consider the Cayley plane, a string
manifold of dimension 16, whose string genus vanishes.  We would like
to know if our genus $\Zsc$ is non-zero for the Cayley plane; it
might give us some information about string cobordism theory.  This
requires computing the function $\zanalog$, an open problem
which may be solvable for a hyperelliptic surface of genus $2$.

\section{The Action of $\Diff_{0}(\Sigma)$ and Torelli}
\label{sec:torelli}

\subsection{The $\Diff_{0}(\Sigma)$ action}
\label{sec:diff0}

The action of $\Diff_{0}(\Sigma)$ on the partition function of a
quantum field theory is well understood.  The seminal work of
Alvarez-Gaum\'{e} and Witten on gravitational anomalies
\cite{AlvarezGaume:1984p1466} initiated the subject.  Gravitational
anomalies are related to $1$-cocycles in the group cohomology of
$\Diff_{0}(\Sigma)$.  For simplicity we assume the partition function
$Z$ only depends on $\Met(\Sigma)$, \emph{i.e.}, $Z: \Met(\Sigma) \to
\bbC$.  If $f \in \Diff_{0}(\Sigma)$ and $\metric \in \Met(\Sigma)$
then the action on a metric is $f: \metric \mapsto (f^{-1})^{*}
\metric$.  The partition function behaves as $Z((f^{-1})^{*} \metric)
= \lambda(f,\metric) Z(\metric)$ where $\lambda: \Diff_{0}(\Sigma)
\times \Met(\Sigma) \to \bbC^{\times}$ is a $1$-cocycle in the group
cohomology of $\Diff_{0}(\Sigma)$:
\begin{equation}
    \lambda( f_{1} \circ f_{0}, \metric) = \lambda(f_{1}, 
    (f_{0}^{-1})^{*}\metric) \; \lambda(f_{0},\metric)
    \text{ where } f_{0},f_{1} \in \Diff_{0}(\Sigma)\,.
    \label{eq:def-cocycle}
\end{equation}
Because of the cocycle condition, the partition function may be
interpreted as a section of a line bundle over $\Teich(\Sigma) =
\Met(\Sigma)/\Diff_{0}(\Sigma)$.  Note that the discussion above is
valid whether or not the metrics have constant curvature.

We make a brief but very important remark before we proceed with
details in the ensuing subsections.  Assume $h_{0},h_{1} \in
\Diff(\Sigma)$ represent the same element in the mapping class group
then there exists $f \in \Diff_{0}(\Sigma)$ such that $h_{1} = f \circ
h_{0}$.  Next we observe that $Z((h_{1}^{-1})^{*}\metric) = Z(
(f^{-1})^{*} \circ (h_{0}^{-1})^{*} \metric) =
\lambda(f,(h_{0}^{-1})^{*} \metric)\; Z((h_{0}^{-1})^{*} \metric)$.
Thus when we work in $\Teich(\Sigma)$ and  want to
understand the action of the mapping class group on the
\emph{partition section} it does not matter which diffeomorphism we
choose as a representative for an element in the mapping class group.

\subsection{The Torelli Action}
\label{sec:torelli-action}

From the definition in the introduction of $\Teichodd(\Sigma)$, the
action of $\Torelli(\Sigma)$ on $\Teich(\Sigma)$ lifts to an action of
$\Torelli(\Sigma)$ on $\Teichodd(\Sigma)$ that we will denote by
$\Torelliodd(\Sigma)$.  We now discuss the action of
$\Torelliodd(\Sigma)$ on the determinant line bundle $\linebundle$ of
\BEG{Theorem~9.1}, where $\linebundle
=((\DET(\dbar_{\delta}))^{*})^{n}$, $\dim M =2n$,  $\dbar_{\delta}:
\Lhz \to \Lho$, and  $\Lhz = \sqrt{K}^{\delta}$ is the square root of the
canonical bundle corresponding the odd spin structure $\delta$.

\begin{lem}\label{lem:foo}
    $\Torelliodd(\Sigma)$ acting on $\Teichodd(\Sigma)$ leaves
    $\linebundle$ invariant.
\end{lem}

We study $\dbar_{\delta}$ in the generic case where it has a one
dimensional kernel.  The determinant line bundle 
$\DET(\dbar_{\delta})$ is isomorphic to the dual lie bundle of a line
subbundle in $\Hoz \subset \Loz$.  Observe that $\Lho$ is the linear
algebraic dual space of $\Lhz$ by using wedge product and integration
over $\Sigma$; we also use the standard inner product on $\Sigma$ to
get an inner product on $\Lhz$.  An elementary computation shows that
if $h_{\delta}$ is in $\ker\dbar_{\delta}$ then its linear algebraic
dual using the inner product is in $\ker \dbar_{\delta}^{*}$.  Thus
the determinant line bundle is the dual line bundle of the line in
$\Hoz$ determined by $h_{\delta}^{2}$.  Moreover,
$\Torelliodd(\Sigma)$ acting on $\Teichodd(\Sigma)$ leaves the
determinant line bundle $\DET(\dbar_{\delta})$ fixed because it sends
the one dimensional kernel of $\dbar_{\delta}$ to the one dimensional
kernel of the transformed $\dbar_{\delta}$.

If $f \in \Diff(\Sigma)$ then $f$
induces a transformation in homology $f_{*}: H_{1}(\Sigma,\bbZ) \to
H_{1}(\Sigma,\bbZ)$.  Define a normal subgroup $\Torellibig(\Sigma) 
\vartriangleleft \Diff(\Sigma)$ 
by
\begin{equation*}
    \Torellibig(\Sigma) = \left\{ f \in \Diff(\Sigma) \;\mid\; 
    f_{*} = \Id \right\}\,.
\end{equation*}
It is  a normal subgroup because if $g \in \Diff(\Sigma)$ then
$(gf g^{-1})_{*} = g_{*}f_{*}g^{-1}_{*} = g_{*}\Id g^{-1}_{*} =\Id$.
Note that $\Torelli(\Sigma) =
\Torellibig(\Sigma)/\Diff_{0}(\Sigma)$ and that $\Sp(2g, \bbZ) =
\Diff(\Sigma)/\Torellibig(\Sigma)$.

Observe that once a standard symplectic basis $(\ba_{i},\bb_{j})$ of
cycles for $H_{1}(\Sigma,\bbZ)$ is chosen then the abelian
differentials $\omega_{i}$ are uniquely determined and depend only on
the homology classes $([\ba_{i}], [\bb_{j}])$.

If $f \in \Diff(\Sigma)$, let $f: (\Sigma,g) \to (\Sigma,\tilde{g})$
where $g$ and $\tilde{g}$ are the respective metrics.  Let $c_{\alpha}
= (\ba_{i},\bb_{j})$ be a choice for the standard cycles on
$(\Sigma,g)$.  A mapping $f$ is of Torelli type if it preserves the
homology.  For such an $f$ there is an induced transformation on
cycles $z$ that gives $f_{*}z = z + \partial \tilde{q}_{\alpha}$.
Hence the Riemann period matrix is invariant:
\begin{equation*}
    \widetilde{\Omega}_{ij} = \int_{\tilde{\bb}_{i}} \tilde{\omega}_{j}
    = \int_{f_{*}\bb_{i}} \tilde{\omega}_{j} = 
    \int_{\bb_{i}}f^{*}\tilde{\omega}_{j} = \Omega_{ij},
\end{equation*}
and moreover $\delta_{ij}= \int_{\ba_{i}}\omega_{j} = 
\int_{\tilde{\ba}_{i}} \tilde{\omega}_{j}$.

\section{Ray-Singer Torsion Revisited}
\label{sec:RS-torsion}

Fix a fiducial metric $\bm{\hat{g}} \in
\Met_{\text{all}}(\Sigma)$, the space of all metrics on $\Sigma$.  The
metric $\hat{\bm{g}}$ determines a complex structure within the
$3(g-1)$ complex dimensional space of complex structures.  We restrict
our discussion to surfaces with genus $g>1$.

For Riemann surfaces, the complex Ray-Singer torsion
theorem~\cite[Theorem~2.1]{Ray-Singer:complex} is a consequence of the
conformal anomaly.  Let $F_{\chi}$ be the flat holomorphic line bundle
associated with the character $\chi: \pi_{1}(\Sigma) \to S^{1}$.
$F_{\chi}$ comes equipped with a hermitian metric that depends only on
the complex structure.  The sections of $K^{n}\otimes \bar{K}^{m}$ are
the ``$(n,m)$-forms'' and are denoted by $T_{n,m}$.  The hermitian
metric on $\Sigma$ allows us to identify $(n,m)$-forms with
$(n-m,0)$-forms.  Let $\dbar_{n}: T_{n,0} \otimes F_{\chi} \to T_{n,1}
\otimes F_{\chi}$ be the basic operator, $\dbar_{n}^{*} : T_{n,1}
\otimes F_{\chi}\to T_{n,0}\otimes F_{\chi}$ be its hermitian adjoint
and let $\Delta_{n}^{(-)} = 2\dbar_{n}^{*}\dbar_{n}$ be the
corresponding laplacian.  Let $\{\phi_{a}\}$ be a basis for $\ker
\dbar_{n}$, the holomorphic sections of $T_{n,0} \otimes F_{\chi}$.
The basis can be chosen to be independent of conformally scaling the
fiducial metric by $e^{2\sigma}$, \emph{i.e.}, it only depends on the
complex structure.  Because of the hermitian metrics on $K$ and
$F_{\chi}$, $\ker \dbar_{n}^{*} \subset T_{n,1} \otimes
F_{\chi} \approx T_{n-1,0} \otimes F_{\chi}$ and may be identified
with the holomorphic sections of $\dbar_{1-n}: T_{1-n,0} \otimes
F_{\chi}^{-1} \to T_{1-n,1} \otimes F_{\chi}^{-1}$.  We also have the
dual space identification
\begin{equation}
    T_{1-n,0} \otimes F_{\chi}^{-1} \approx \left( T_{n-1,0} \otimes
    F_{\chi} \right)^{*}\,.
    \label{eq:dual-spaces}
\end{equation}
Let $\{\psi_{\alpha}\}$ be the holomorphic sections of $T_{1-n,0}
\otimes F_{\chi}^{-1}$ (chosen to be independent of the conformal
factor $\sigma$).  The conformal anomaly implies that under an
infinitesimal conformal change of the metric $\bm{\hat{g}} \to
e^{2(\delta\sigma)}\;\bm{\hat{g}}$ we have~\cite{Alvarez:1982zi}
\begin{equation}
    \delta_{\sigma}\; \log\left(\frac{\detp \Delta_{n}^{(-)}}{\det
    \langle \psi_{\alpha},\psi_{\beta}\rangle \; \det \langle
    \phi_{a}, \phi_{b} \rangle} \right)_{\chi} = -
    \frac{1+6n(n-1)}{6\pi} \int_{\Sigma}d^{2}z\;
    \sqrt{\bm{\hat{g}}}\;\hat{R}\;(\delta\sigma)\,.
    \label{eq:conf-anomaly}
\end{equation}
The term inside the parentheses on the left hand side of the
equation is the Quillen metric of the determinant line bundle
$\DET(\dbar_{n})$.  The determinant term associated with $\ker
\dbar_{n}^{*}$ appears in the denominator because of the dual space
identification given in eq.~\eqref{eq:dual-spaces}.  This formula is
valid for $2n \in \bbZ$.  Note that the right hand side is independent
of $F_{\chi}$.

The Ray-Singer torsion results correspond\footnote{The $T_{0} = T_{1}$
result of Ray and Singer is related to the two laplacians we can
define.} to the case $n=0$.  Let $\chi$ and $\chi'$ be two non-trivial
characters.  For both characters, $\ker\dbar_{0} = \{0\}$ and
$\dim\ker\dbar_{0}^{*} = g-1$.  Also the metric on $T_{1,0} \otimes
F_{\chi}^{-1}$ is independent of the conformal factor and therefore
the term $\det \langle \psi_{\alpha},\psi_{\beta}\rangle$ in the left
hand side of \eqref{eq:conf-anomaly} does not change under a conformal
variation.  Putting all this information together gives
\begin{equation}
    \delta_{\sigma}\; \left( \log\det \Delta_{0}^{(-)}(\chi) - \log 
    \Delta_{0}^{(-)}(\chi') \right) =0\,.
    \label{eq:R-S}
\end{equation}
This is the Ray-Singer result for complex analytic torsion on
Riemann surfaces.  It says that the ratio $(\det
\Delta_{0}^{(-)}(\chi))/(\det \Delta_{0}^{(-)}(\chi'))$ only depends
on the complex structure. 

An immediate consequence of \eqref{eq:conf-anomaly} is
\begin{thm}[Generalized Ray-Singer Torsion on Riemann Surfaces]
    \label{thm:R-S}
    Consider a collection $\{(n_{r},\chi_{r},k_{r})\}_{r=1}^{N}$ where
    $2n_{r} \in \bbZ$, $\chi_{r}: \pi_{1}(\Sigma) \to S^{1}$ is a
    character, and $k_{r} \in \bbZ$.  If this collection satisfies
    \begin{equation*}
        \sum_{r=1}^{N} k_{r} \left[ 1 + 6 n_{r}(n_{r}-1)\right] =0
    \end{equation*}
    then
    \begin{equation*}
	\sum_{r=1}^{N} k_{r} \log\left(\frac{\detp \Delta^{(-)}}{\det
	\langle \psi_{\alpha},\psi_{\beta}\rangle \; \det \langle
	\phi_{a}, \phi_{b} \rangle} \right)_{n_{r},\chi_{r}}
    \end{equation*}
    only depends on the complex structure and is independent of the
    choice of hermitian metric on $\Sigma$.
\end{thm}

The two best known examples of this theorem in string theory are the
$26$ dimensional bosonic string~\cite{Polyakov:1981rd} with $\{
(0,1,26/2), (-1,1,-1)\}$, and the $10$ dimensional
superstring~\cite{Polyakov:1981re} with collection
\begin{equation*}
    \bigl\{ (0,1,10/2), (1/2,\chi,-10/2), (-1,1,-1),
    (-1/2,\chi^{-1},1) \bigr\}
\end{equation*}
associated with $10$ bosons, $10$ Majorana fermions, diffeomorphisms
(vector fields), super-diffeomorphisms (square root of vector fields).
In the above $\chi$ can be any character corresponding
to a spin structure.

\subsection{Ray-Singer Torsion for Spinors}
\label{sec:one-half}

We can be very explicit in general genus in the case of spinors.  Pick
a reference point $P_{0} \in \Sigma$ and in the standard fashion
identify $\Jacz$ with $\Jacgo$ as discussed in
\ref{sec:app-kappa}.  The character $\chi$ corresponds to a
point $u \in \Jacz$.  There is a special spin structure $\spinbasic$ such
that the determinant of the laplacian acting on $\spinbasic \otimes
F_{\chi}$ is given by \eqref{eq:det-DD} where $u \in \Jacz$ is the
point corresponding to $F_{\chi}$.  This example shows
that $(\det \Delta(u))/(\det \Delta(u'))$ only depends on the complex
structure in agreement with the generalized Ray-Singer theorem.

We introduce the notation
\begin{equation}
    Q(u) = \frac{\detp D(u)^{*}D(u)}{\det 
    \langle \psi_{\alpha},\psi_{\beta}\rangle \; 
    \det \langle \phi_{a}, \phi_{b} \rangle}
    \label{eq:Q-metric}
\end{equation}
for convenience and define  subvarieties $\cV_{0}, \cV_{1},
\cV_{2}, \dotsc$ of $\Jacz$ where $\cV_{k}$ is the set of points $u\in
\Jacz$ where $\dim \ker D(u) = k$.  Note that the theta divisor is
given by $\Theta = \cup_{k=1}^{\infty} \cV_{k}$.  Because the Dirac
operator has index zero, the matrices $\langle
\psi_{\alpha},\psi_{\beta}\rangle$ and $\langle \phi_{a}, \phi_{b}
\rangle$ are the same size.  We have seen that if $u,u'
\in \cV_{0}$ then $Q(u)/Q(u')$ is independent of the choice of
hermitian metric on $\Sigma$.  Similarly,  if $u \in
\cV_{k}$ and $v \in \cV_{l}$ then $Q(u)/Q(v)$ will be independent of
the choice of hermitian metric on $\Sigma$.  

In general, $\phi_{a}$ is a holomorphic section of $\spinbasic \otimes
F_{\chi}$ and $\psi_{\alpha}$ is a holomorphic section of $\spinbasic
\otimes F_{\chi}^{-1}$.  When $F_{\chi}$
corresponds to a semi-characteristic so that $\spinbasic \otimes
F_{\chi}$ is a spin structure,  $F_{\chi}^{2}$ is the
trivial bundle.  Thus $F_{\chi} \approx F_{\chi}^{-1}$ and
we can identify the determinants in the denominator of
\eqref{eq:Q-metric}.  In the generic case where all the odd spin
structure are in $\cV_{1}$ we have the explicit computations
\eqref{eq:det-prime} (an example involving $\cV_{0}$ and $\cV_{1}$),
and \eqref{eq:det-univ} (an example only involving $\cV_{1}$).

\section{The Bosonic Determinant}
\label{sec:bosonic-det}

We rewrite our key formula \BEG{(4.12)} by changing
the orientation of the surface $z \to \zbar$.  Now the
semiclassical partition ``function'' (section) becomes
\begin{equation}
    \Zsc = (\vol \Sigma)^{n}
    \left(\frac{\detp\partial_{\delta}}{\Nb_{\delta}^{2}} \right)^{n}
    \int_{M} \det\left[D \left(\partial_{0} \otimes
    I_{2n}\right)\right]^{-1/2}\,,
    \label{eq:def-Zsc}
\end{equation}
here $\dim M = 2n$, $D = *i(\dbar \otimes I_{2n} + \sfA^{0,1})$ acting on
$\Loz \otimes X^{*}(TM)$ with $X:\Sigma \to M$  a constant map.  We
integrate over the space of constant maps $M$.  We assume the odd spin
structure $\delta$ is generic, so $\ker \partial_{\delta}$ is
$1$-dimensional and  chose a non-zero element in the kernel
$\hb_{\delta} \in \Lzh$ and also $\Nb_{\delta} \in \bbC$ such
that $\hb_{\delta}/\Nb_{\delta}$ has norm $1$.  Hence
$\hb_{\delta}/\Nb_{\delta}$ is an element of norm $1$ in $\Lzh$.  The
term $\sfA^{0,1}$ in the definition of $D$ is $\hb_{\delta}^{2}
\otimes \mathcal{R}/2\pi$ with $\mathcal{R}$ the curvature $2$-form 
of $M$ pulled back via the constant map $X$.

$\Zsc$ depends on a metric $\metric$ on $\Sigma$ and an odd spin
structure $\delta$.  In the Section~\ref{sec:quarfs} we review how a
choice of symplectic basis $b$ for $H_{1}(\Sigma,\bbZ)$ fixes an even
spin structure $\sqrt{K}^{b}$.  Adding an appropriate element $w \in
H^{1}(\Sigma, \bbZ_{2})$ gives an odd spin structure so $\Zsc =
\Zsc(\metric,b,w)$.  If $f \in \Diff(\Sigma)$ then $f$ induces a map
$(\metric;b,w) \mapsto \left( (f^{-1})^{*}\metric; f_{*}b,
(f^{-1})^{*}w \right)$.

In BEG we showed that
\begin{equation}
    \Zsc = \left( \frac{\vol \Sigma}{\detz} \right)^{n}
    \left(\frac{\detp\partial_{\delta}}{\Nb_{\delta}^{2}}
    \right)^{n} \int_{M} \prod_{k=1}^{n} 
    \frac{\zanalog\left(x_{k}
    z(\hb_{\delta}^{2})\right)}{\vartheta[\kappa]\left(x_{k}
    z(\hb_{\delta}^{2})\right)}\,.
    \label{eq:def-Zsc-1}
\end{equation}
See the Introduction for the definitions of the terms except for
$\zanalog$ that can be found in \ref{sec:app-kappa}.  Riemann
surface theory \cite{Mumford:1} gives us an explicit expression for
the square of the spinor:
\begin{equation}
    h_{\delta}^{2} = \sum_{k}\frac{\partial
    \vartheta[\delta](0)}{\partial z^{k}} \omega_{k}\,.
    \label{eq:def-hdelta2}
\end{equation}

Formula \eqref{eq:def-Zsc} contains a specific $1$-form with curvature
zero. The flatness arises because the $1$-form is
the product of the pullback of the curvature on the target space $M$
by the constant map, and the square of the anti-holomorphic
spinor.

It is useful to consider the family of operators $D = i * (\dbar +
A^{0,1})$ with $\Azo$ a flat connection.  Here $D: \Loz \to
\Lzz$ is parametrized by $\Jacz$.  The determinant line bundle
$\cL = \DET D \to \Jacz$ has a Quillen
hermitian metric with connection $\nu$ and curvature $d\nu$ given by
the standard translationally invariant polarization form on $\Jacz$. 
It also has a unique holomorphic cross section (up to scale), see 
\BEG{Appendix C}.

The computation of the bosonic determinant in \eqref{eq:def-Zsc} 
involves three steps.
\begin{enumerate}
    
    \item In Section~\ref{sec:trivi-det} we will lift $\cL$ and its
    holomorphic cross section to the covering space $\Hzo$ of $\Jacz$,
    then trivialize the lift hence making the cross section a function
    which we will identify as a $\vartheta$-function.
    
    \item In \ref{sec:app-kappa}, we use elliptic analysis
    to study $\cL$ and its unique holomorphic section.
    
    \item Combining the two previous items leads to a formula for the
    aforementioned determinant after studying a PDE as discussed in
    \BEG{Section 6}.
\end{enumerate}

\subsection{Trivializing the Determinant Line Bundle}
\label{sec:trivi-det}

Remember\footnote{The content in this section is required for the flow
of this paper; there is overlap with work in BEG.} that the
jacobian is defined by $\Jacz = \Hzo/L_{\Omega}$.  We identify $\Hzo$
with $\bbC^{g}$ by choosing a basis of $\Hzo$ given by formula
\eqref{eq:z-coord} with $z_{j}$ the coordinates of $\bbC^{g}$.  With
this convention the lattice $L_{\Omega} \subset \Hzo$ is given by
\eqref{eq:B-coord}.
\begin{align}
    \Azo & = 2\pi i \sum
    z_j(\Omega-\bar{\Omega})^{-1}_{jk}\bar{\omega}_k\;,
    \label{eq:z-coord}\\
    \Bzo_{nm} & = 2\pi i \sum (m + \Omega n)_j
    (\Omega-\bar{\Omega})^{-1}_{jk} \bar{\omega}_k, \text{ where }
    m,n \in \bbZ^{g}\;.
    \label{eq:B-coord}
\end{align}
With these conventions, the quasiperiodicity of the theta function 
are associated with $z \to z + m + \Omega n$.

The jacobian torus defined above is the dual torus to the one normally
used by algebraic geometers.  If $(\alpha_{j}, \beta_{k})$ is a
symplectic basis for $H^{1}(\Sigma,\bbR)$ in terms of harmonic
$1$-forms then the abelian differentials are given by $\omega_{j} =
\alpha_{j} + \sum_{k} \Omega_{jk}\beta_{k}$.  The algebro-geometric
jacobian is $\Hoz$ modulo the integer lattice $H^{1}(\Sigma,\bbZ)$,
see \ref{sec:app-kappa}.  We know that the linear algebraic
dual of $\Hzo$ is $\Hoz$.  There is a hermitian inner product on
$\Hoz$ and therefore there is a conjugate linear isomorphism with
$\Hzo$ that identifies the basis vectors $\omega_{i}$ with
$\sum_{k}(\Omega - \bar{\Omega})_{jk}\bar{\omega}_{k}$ up to an
overall normalization.

Let $\pi: \Hzo \to \Jacz$ be the standard projection and define the
pull back line bundle $\Ltil = \pi^{*} \mathcal{L} \to
\Hzo$ with connection $\nutil = \pi^{*}\nu$. On $\Hzo$ the
$1$-form
\begin{equation}
    \rho = \frac{i}{2\pi} \int_{\Sigma} \overline{\Azo} \wedge 
    d_{\Hzo}\Azo
    \label{eq:def-rho}
\end{equation}
has the property that $0 = d(\nutil-\rho) = \dbar(\nutil -\rho)$.  In
\BEG{Section~6} we used the flat connection $\nutil -\rho$ to
trivialize $\Ltil$.  We briefly review a slight modification of that
discussion here.  Let $\Ltil_{0}$ be the fiber over $0 \in \Hzo$.  To
identify the fiber $\Ltil_{0}$ with $\bbC$ we choose an arbitrary
non-zero point $\hat{\sigma}_{0} \in \Ltil_{0}$.  Given two points
$A_{0}, A_{1} \in \Hzo$, an integral $\int_{A_{0}}^{A_{1}} \cdots$ is
always taken along the straight segment\footnote{Many of our integrals
involve a flat connection so the choice of integration path is
irrelevant.  It will matter to us when we project the curve down to
the $\Jacz$ and try to interpret results geometrically.} joining the
two points.  We defined the \emph{flat trivialization} $\varphi: \Ltil
\to \Hzo \times \bbC$ by using the flat connection $\nutil -\rho$ to
give us a holomorphic trivialization.  More explicitly, let $\sigma$
be a point in the fiber of $\Ltil$ over $A$ then parallel transport
$\sigma$ along the straight segment from $A$ to $0$ to obtain a point
$\sigma_{0}\in \Ltil_{0}$.  The trivialization map is $\varphi: \sigma
\mapsto (A,\sigma_{0}/\hat{\sigma}_{0})$.  Abusing notation slightly,
we write
\begin{equation}
	\sigma(A) = (\varphi\sigma)(A) \;
	\exp\!\left(-\int_0^A(\nutil-\rho)\right) \hat{\sigma}_{0}\;.
	\label{eq:s-par}
\end{equation}
To make things more standard we defined a slightly different 
trivialization that we called the \emph{standard trivialization} 
$\Phi$ by 
multiplying the above by a non-vanishing holomorphic function on 
$\bbC^{g}$. The standard trivialization is defined by
\begin{equation}
	(\Phi\sigma)(A) = \exp\!\left(-\pi i \sum_{j,k} z_j
	(\Omega-\bar{\Omega})^{-1}_{jk}
	z_k\right)(\varphi\sigma)(A)\;.
	\label{eq:std-cocycle}
\end{equation}
If $s$ is any section of $\cL$ and if $\tilde{s} = \pi^{*}s$ is the
pull back section to $\Ltil$ then $\tilde{s}(A+B) = \tilde{s}(A)$ for
any lattice vector $B$.  Consequently the pull back of any section of
$\cL$ in the trivialized bundle (trivialized by the standard
trivialization) is represented by a function $\Phi\tilde{s}$ with
quasi-periodicity properties
\begin{equation}
    (\Phi\tilde{s})(A+B)  =  \chi(B_{nm})
      \times  e^{-\pi i \sum_{j,k} n_j \Omega_{jk} n_k}
     e^{-2\pi i \sum n_j z_j}\; (\Phi\tilde{s})(A)\;.
    \label{eq:theta-cocycle}
\end{equation}
The above is the standard transformation law for a theta function 
with lattice character
\begin{align}
    \chi(B_{nm}) &= e^{-\pi i \sum m_j n_j}e^{\int_0^B \nu}\\
    & =  e^{-\pi i \sum m_j n_j} \,\hol(\gamma^0_{nm})^{-1} \;,
    \label{eq:def-chi}
\end{align}
where $\hol(\gamma^0_{nm})$ is the holonomy of the Quillen connection.
The closed curve $\gamma^0_{nm}$ is the projection into the jacobian
of the straight segment from $0$ to $B_{nm}$ in $\Hzo$.  The bundle
$\cL \to \Jacz$ has a unique holomorphic cross section (up to scale).
If we write
\begin{equation}
    \chi(B_{nm}) = e^{-2\pi i n\cdot b}e^{2\pi i m\cdot a}
    \label{eq:def-char}
\end{equation}
then (\ref{eq:theta-cocycle}) is the transformation law for the
function $\vartheta\left[\begin{smallmatrix} a \\ b \end{smallmatrix}
\right](z)$.  In \ref{sec:app-kappa} we show that the
characteristic $\left[\begin{smallmatrix} a \\ b \end{smallmatrix}
\right]$ associated with the holomorphic section of $\mathcal{L} =
\DET(D) \to \Jacz$ is $\kappa= \Omega a + b$ where
$\kappa$ is the vector of Riemann constants.

The line bundle $\mathcal{L} = \DET(D) \to \Jacz$ has a unique
holomorphic section $\theta_{\kappa}$ up to scale.  The results above
state that the pullback section $\pi^{*}\theta_{\kappa}$ on
$\Ltil=\pi^{*}\cL$ is related to the $\vartheta$-function on the
trivialized bundle $\Hzo \times \bbC$ by
\begin{equation}
    (\pi^{*}\theta_{\kappa})(A) = \exp\!\left(+\pi i \sum_{j,k} z_j
	(\Omega-\bar{\Omega})^{-1}_{jk}
	z_k\right) \vartheta[\kappa](z)\;
	\exp\!\left(-\int_0^A(\nutil-\rho)\right) 
	\tilde{\sigma}_{0}\,.
    \label{eq:theta-hzo}
\end{equation}
It is explicit from the above that the lift of the divisor
$\Theta_{\kappa}$ of $\theta_{\kappa}$ to $\Hzo$ is the same as the
zero set of $\vartheta[\kappa](\cdot)$.

\section{Symplectic Action}
\label{sec:symplectic}

\subsection{Facts About the Symplectic Group}
\label{sec:symplectic-facts}

Let $\Lambda = \bigl( \begin{smallmatrix} A & B \\ C & D
\end{smallmatrix} \bigr) \in \Sp(2g, \mathbb{F})$ for some field 
$\mathbb{F}$. Then
\begin{equation}
    \Lambda^{t} \; \sympmetric \Lambda = \sympmetric,
    \label{eq:def-symp-matrix}
\end{equation}
so that
\begin{equation*}
    D^{t} A - B^{t}C =I, \quad A^{t}C-C^{t}A=0, \quad B^{t}D - 
    D^{t}B=0\,.
\end{equation*}
Thus $A^{t}C$ and $B^{t}D$ are symmetric matrices and
\begin{equation}
    \begin{pmatrix}    
	A & B \\
        C & D
    \end{pmatrix}^{-1} =
    \begin{pmatrix}    
	D^{t}& -B^{t} \\
        -C^{t} & A^{t}
    \end{pmatrix}
    \label{eq:left-inv}
\end{equation}
is a left inverse hence
\begin{equation*}
    \begin{pmatrix}
        I & 0 \\
	0 & I 
    \end{pmatrix}
    =
    \begin{pmatrix}    
	A & B \\
	C & D
    \end{pmatrix}
    \begin{pmatrix}    
	D^{t}& -B^{t} \\
	-C^{t} & A^{t}
    \end{pmatrix}
    = \begin{pmatrix}
        AD^{t} -BC^{t} & -AB^{t} + BA^{t} \\
	CD^{t} - DC^{t} & -CB^{t} + DA^{t}
    \end{pmatrix},
\end{equation*}
implying
\begin{subequations}\label{eq:symp-constraints}
\begin{align}
	D^{t} A - B^{t}C &=I, & A^{t}C-C^{t}A & =0, & B^{t}D - 
	D^{t}B &=0\,, 
	\label{eq:symp-constraints-0} \\
	AD^{t} - BC^{t} &= I, & AB^{t} - BA^{t} &=0, & CD^{t} - DC^{t} 
	&= 0 \,.
	\label{eq:symp-constraints-1}
\end{align}
\end{subequations}
Equations \eqref{eq:symp-constraints} completely characterize a 
symplectic matrix.

Equation \eqref{eq:symp-constraints-0}
implies \eqref{eq:symp-constraints-1}:  multiply the first of
\eqref{eq:symp-constraints} on the left by $A$
\begin{align*}
    0 = -A + AD^{t}A - AB^{t}C &= (-I + AD^{t} - BC^{t})A + 
    B \underbrace{C^{t}A}_{A^{t}C} - 
    AB^{t}C \\
    &= (-I + AD^{t} - BC^{t})A + (BA^{t} -AB^{t})C\,.
\end{align*}
Multiply the first of 
\eqref{eq:symp-constraints} on the right by $D^{t}$
\begin{align*}
    0 = -D^{t} + D^{t}AD^{t} - B^{t}CD^{t} &= D^{t}(-I + AD^{t} -BC^{t}) 
    + \underbrace{D^{t}B}_{B^{t}D}C^{t} - B^{t}CD^{t} \\
    & =  D^{t}(-I + AD^{t} -BC^{t}) + B^{t}(DC^{t} - CD^{t})\,.
\end{align*}
These equations may be written in matrix form (after transposing the 
second group) as
\begin{equation*}
    \begin{pmatrix}    
	-I + AD^{t} - BC^{t} & BA^{t} -AB^{t} \\
        CD^{t} - DC^{t} & -I + DA^{t} -CB^{t}
    \end{pmatrix}
    \begin{pmatrix}
        A & B \\
	C & D
    \end{pmatrix}
    =0\,.
\end{equation*}
Since the matrix on the right is invertible the one on left must be
the zero matrix giving \eqref{eq:symp-constraints-1}.

\subsection{Transformation Laws}
\label{sec:general}

Pick a symplectic basis $(\ba_{j},\bb_{k})$ for $H_{1}(\Sigma,\bbZ)$
and represent the dual basis via harmonic differentials
$(\alpha_{j},\beta_{k})$ .  The standard normalized holomorphic
abelian differentials are given by $\omega_{j} = \alpha_{j} + \sum_{k}
\Omega_{jk}\beta_{k}$ where $\Omega$ is the period matrix.   The
action of $\Lambda = \bigl( \begin{smallmatrix} A & B \\ C & D
\end{smallmatrix} \bigr) \in \Sp(2g, \bbZ)$ gives new 
bases\footnote{We use the notation that a primed quantity is the 
symplectic transform of the unprimed quantity.}
$(\ba',\bb')$ and $(\alpha', \beta')$:
\begin{equation}
    \begin{pmatrix}
        \ba'  \\
        \bb'
    \end{pmatrix}
    = \begin{pmatrix}
        D & C  \\
        B & A  
    \end{pmatrix}
    \begin{pmatrix}
        \ba  \\
        \bb  
    \end{pmatrix}
    \quad\text{and}\quad
    \begin{pmatrix}
	\alpha  &
	\beta
    \end{pmatrix}
	= 
    \begin{pmatrix}
	\alpha' & \beta'
    \end{pmatrix}
    \begin{pmatrix}
	D & C  \\
	B & A  
    \end{pmatrix}.
    \label{eq:symp-ab}
\end{equation}
Notice that
\begin{equation}
    \begin{pmatrix}
        D & C  \\
        B & A  
    \end{pmatrix}^{-1}
    = \begin{pmatrix}
        A^{t} & -C^{t}  \\
        -B^{t} & D^{t}  
    \end{pmatrix} .
    \label{eq:symp-inverse}
\end{equation}
Under the symplectic action
\begin{equation}
    \omega = \omega' (C\Omega + D) \quad\text{and}\quad \Omega' =
    (A\Omega + B) (C \Omega + D)^{-1}\,.
    \label{eq:symp-omega}
\end{equation}
where $\omega$ is the row vector $(\omega_{1},\dotsc, \omega_{g})$.
The covering space of the jacobian torus $\Jacz$ is $\bbC^{g}$.  In
terms of standard coordinates on $\bbC^{g}$, the jacobian torus is
given by the identifications $z \sim z + m + n\Omega$ where we write
the coordinates as a row vector.  Under the action of $\Lambda$, the
transformed torus is described by $z' \sim z' + m' + n'\Omega'$ so
that
\begin{equation}
    z = z' (C\Omega + D)\,.
    \label{eq:symp-z}
\end{equation}
Note the $\Sp(2g, \bbZ)$ invariance:
\begin{equation}
    \sum_{j} \omega_{j} \; \frac{\partial}{\partial z_{j}} = \sum_{j}
    \omega'_{j} \; \frac{\partial}{\partial z'_{j}}\,.
    \label{eq:symp-der}
\end{equation}

According to Fay~\cite{Fay:theta,Alvarez-Gaume:1986es}
a theta function with generic characteristics
transforms under symplectic transformations by
\begin{equation}
    \thchar{a'}{b'}(z'; \Omega') = \varepsilon(\Lambda) e^{-i\pi 
    \phi(a,b,\Lambda)} e^{\sum Q_{ij}z_{i}z_{j}} \; \det(C\Omega +D)^{1/2}
    \; \thchar{a}{b}(z;\Omega)\,,
    \label{eq:symp-theta}
\end{equation}
where $\varepsilon: \Sp(2g, \bbZ) \to \bbZ_{8}$ is a phase
independent of $z$ and $\Omega$,
\begin{equation}
    \phi(a,b,\Lambda) = aD^{t}B a + b C^{t}Ab - \left[ 2 a B^{t}C b + (aD^{t} 
    - bC^{t}) (AB^{t})_{d} \right],
    \label{eq:symp-phi}
\end{equation}
and
\begin{equation}
    \begin{pmatrix}
        a'  \\
        b'  
    \end{pmatrix}
    = \begin{pmatrix}
        D & -C  \\
        -B & A  
    \end{pmatrix}
    \begin{pmatrix}
        a  \\
        b  
    \end{pmatrix}
    + \frac{1}{2}
    \begin{pmatrix}
        (CD^{t})_{d}  \\
        (AB^{t})_{d}  
    \end{pmatrix}.
    \label{eq:symp-char}
\end{equation}
In the above $(AB^{t})_{d}$ means the diagonal entries of the matrix
product as a column vector.  We do not need an explicit form for the
symmetric matrix $Q_{ij}$ because the condition
$\tfrac{1}{2}\pont_1(M)=0$ eliminates that term in our computations.
Later we show that transformation law \eqref{eq:symp-char} is
equivalent to the transformation law for quarfs \eqref{eq:w-transf}.

Let $\delta$ be the odd theta characteristic for a spin structure with
holomorphic spinor $h_{\delta}$.  As noted in the introduction
\begin{equation}
    h_{\delta}^{2} = \sum_{k} \frac{\partial 
    \vartheta[\delta]}{\partial z_{k}}(0;\Omega)\; \omega_{k}\,.
    \label{eq:hdelta}
\end{equation}
$\vartheta[\delta](0;\Omega)=0$,
\eqref{eq:symp-theta} and \eqref{eq:symp-der} implies
\begin{equation}
    (h'_{\delta'})^{2} = \varepsilon(\Lambda) e^{-i\pi 
    \phi(\delta, \Lambda)}\; \det(C\Omega + D)^{1/2} \; h_{\delta}^{2}\,.
    \label{eq:symp-hdelta}
\end{equation}

Using
\begin{equation*}
    \int_{\Sigma} \bar{\omega}_{k} \wedge \omega_{j} = (\Omega - 
    \bar{\Omega})_{kj}
\end{equation*}
and \eqref{eq:z-coord}
implies that if $\Azo$ is a $1$-form in $\Hzo$ then
\begin{equation}
    2\pi i\, z_{k}(\Azo) = \int_{\Sigma} \Azo \wedge \omega_{k}\,.
    \label{eq:coord-A}
\end{equation}
If $H(\Omega)$ is the Hodge matrix associated with period matrix
$\Omega$, then
\begin{equation}
    \left\lvert \det(C\Omega + D) \right\rvert^{2} \det H(\Omega') = 
    \det H(\Omega)\,.
    \label{eq:symp-hodge}
\end{equation}

In our applications the object that enters is not the holomorphic
section of $\sqrt{K}^{\delta}$ but the anti-holomorphic section of
$\sqrt{\overline{K}}^{\delta}$ which we denote by $\hb_{\delta}$.  The
standard coordinate for $\hb_{\delta}^{2} \in \Hzo$ is 
\begin{equation}
    2\pi i\, z_{k}(\hb_{\delta}^{2})  = \int_{\Sigma} \hb_{\delta}^{2} \wedge
    \omega_{k} = \sum_{j} \overline{ \frac{\partial 
    \vartheta[\kappa]}{\partial z_{j}}(0;\Omega)} \;
    (\Omega-\bar{\Omega})_{jk}\,.
    \label{eq:coord-h2}
\end{equation}
Using the complex conjugate of \eqref{eq:symp-hdelta} gives
\begin{align*}
    2\pi i \, z'\left((\hb'_{\delta'})^{2}\right) (C\Omega + D) & =
    \left( \int_{\Sigma} (\hb'_{\delta'})^{2} \wedge \omega' \right)
    (C\Omega +D)\,, \\
     & = \overline{\varepsilon(\Lambda)} e^{+i\pi \phi(\delta)}\; 
     \det(C\bar{\Omega} +
     D)^{1/2} \int_{\Sigma} (\hb_{\delta})^{2} \wedge \omega\,, \\
     & = \overline{\varepsilon(\Lambda)} \; e^{+i\pi \phi(\delta)}\;
     \det(C \bar{\Omega} + D)^{1/2} \; 2\pi i \, 
     z(\hb_{\delta}^{2})\,.
\end{align*}
Thus the standard coordinates for
$(\hb'_{\delta'})^{2}$ and the standard coordinates for $\hb_{\delta}^{2}$
are related by
\begin{equation}
    z'((\hb'_{\delta'})^{2})(C\Omega + D) = \overline{\varepsilon(\Lambda)} \;
    e^{+i\pi \phi(\delta)}\; \det(C \bar{\Omega} + D)^{1/2} \;
    z(\hb_{\delta}^{2})\,.
    \label{eq:symp-hp-h}
\end{equation}
This differs by a scale from the standard relationship 
\eqref{eq:symp-z} for the coordinates between corresponding 
points in $\Hzo$ under the action of $\Sp(2g, \bbZ)$.

To work out the transformation properties of
$\zanalog/\vartheta[\kappa](\cdot)$ we use: (1) the relation
\eqref{eq:z-theta}, (2) Remark~\ref{rem:z-tranf}, (3) $M$ is a string
manifold, (4) equation~\eqref{eq:def-Zsc-1} has a factor that involves
$\int_{M}$.  Putting all these observations together gives
\begin{thm}
    \label{thm:transf-integral}
    Let $M$ be a string manifold  with
    $\dim M=2n$, if $\Lambda \in \Sp(2g,\bbZ)$ is represented by $f
    \in \Diff(\Sigma)$, then under the action of $f$ we have
\begin{align}
    \int_{M} \prod_{r=1}^{n} \frac{\zanalogp\left(x_{r}\; z'\left(
    (\hb'_{\delta'})^{2} \right); \Omega'
    \right)}{\vartheta[\kappa']\left(x_{r}\; z'\left(
    (\hb'_{\delta'})^{2} \right); \Omega' \right)} & =
    {\varepsilon(\Lambda)}^{-n} \; e^{+i\pi n \phi(\delta)}\; \det(C
    \bar{\Omega} + D)^{n/2} \nonumber \\
     & \quad \times \int_{M} \prod_{r=1}^{n}
     \frac{\zanalog\left(x_{r}\; z\left( \hb_{\delta}^{2}
     \right); \Omega \right)}{\vartheta[\kappa]\left(x_{r}\; z\left(
     \hb_{\delta}^{2} \right); \Omega \right)}\;.
    \label{eq:transf-theta}
\end{align}
\end{thm}

For a generic odd spin structure $\delta$, the determinant line
bundle $\DET(\partial_{\delta})$ is the line bundle dual to the
complex line bundle generated by $\hb_{\delta}^{2}$ in $\Hzo$.
Equation~\eqref{eq:symp-hdelta} tells us how the line transforms under
symplectic transformations thus we see that \eqref{eq:transf-theta} is
the transformation law for a section of the determinant line bundle
$\DET(\partial_{\delta})^{n}$ in the trivialization given by
$(1/\hb_{\delta}^{2})^{n}$ hence
\begin{thm}
    \label{thm:inv-section}
    The function
    \begin{equation}
	 \int_{M} \prod_{r=1}^{n} \frac{\zanalog\left(x_{r}\;
	 z\left( \hb_{\delta}^{2} \right); \Omega
	 \right)}{\vartheta[\kappa]\left(x_{r}\; z\left(
	 \hb_{\delta}^{2} \right); \Omega \right)} : 
	 \Teichodd(\Sigma) \to \bbC
	\label{eq:inv-section}
    \end{equation}
    represents a section of the line bundle
    $\DET(\partial_{\delta})^{n} \to \Mododd(\Sigma)$.
\end{thm} 

A drawback of the two theorems above is that the integrals are not
holomorphic.  They will be reformulated shortly to behave in a more
holomorphic form analogous to the genus 1 case.  To do this we have to
understand the transformation of fermion determinants under the
geometric symplectic action.  Much of our intuition draws from the
genus $1$ case and \ref{sec:genus-1} is devoted to it and the
insights it gives for higher genus.  We hope this review of the genus
one case will be illuminating.

\section{Quadratic Refinements and Spin Structures}
\label{sec:quarfs}

\subsection{Introduction}
\label{sec:quarf-intro}

A choice of symplectic basis for $H_{1}(\Sigma,\bbZ)$ gives a period
matrix $\Omega$, normalized holomorphic abelian differentials, and a
spin structure $\sqrt{K}$; see \ref{sec:app-kappa}.  The group
$\Sp(2g,\bbZ)$ operates on $H_{1}(\Sigma,\bbZ)$ and
$H^{1}(\Sigma,\bbZ_{2})$ and so induces an operation on spin
structures.  To see this action explicitly on $\Teichspin(\Sigma)$ we
introduce the space of quadratic refinements of cup product on
$H^{1}(\Sigma,\bbZ_{2}) \to H^{2}(\Sigma,\bbZ_{2}) \approx \bbZ_{2}$
which are called \emph{quarfs}.  Though these objects are well known,
we give an expository account of them.  Following
Atiyah~\cite{Atiyah:spin}, we exhibit a map from spin structures to
quarfs.  Both spaces are principal homogeneous spaces of
$H^{1}(\Sigma,\bbZ_{2})$ and the map commutes with the action of
$H^{1}(\Sigma,\bbZ_{2})$.  Quarfs will help us compute how
$\Sp(2g,\bbZ)$ acts on $\Teichspin(\Sigma)$.

\subsection{Quarf Primer}
\label{sec:quarf-primer}

Let $V$ be a $\bbZ_{2}$ vector space and let $b$ be a nondegenerate
bilinear form on $V$.  A quadratic refinement (quarf) of $b$ is a
function $q: V \to \bbZ_{2}$ with the property that
\begin{equation}
    q(v + w) - q(v) -q(w) +q(0) = b(v,w) \,.
    \label{eq:def-quarf}
\end{equation}
The definition implies that $b$ is a symmetric bilinear form and
$b(v,v)=0$.

Let $T_{z}$ be translation by $z \in V$ and let $q_{z} = q \circ
T_{z}$.  $q_{z}$ is a quarf:
\begin{align*}
    q_{z}(v + w) + q_{z}(v) +q_{z}(w) +q_{z}(0) & = q(v+w+z) +q(v+z) 
    + q(w+z) +q(z)\,,  \\
     & = \underbrace{q(v + w+z) + q(v) + q(w+z) + q(0)}_{b(v,w+z)} \\
     & \qquad\qquad + q(v) + \underbrace{q(w+z)}_{\text{sum cancel}} + q(0) \\
     &\quad +q(v+z) + \underbrace{q(w+z)}_{\text{sum cancel}} +q(z) \,, \\
     & = b(v,w+z) + b(v,z) \,,\\
     &= b(v,w)\,.
\end{align*}
Also if $\tilde{q}_{r}(v) = q(v) + r$ with $r \in \bbZ_{2}$ then
$\tilde{q}_{r}$ is a quarf because $4r=0$.

Let $Q_{0}$ be the quarfs with $q(0)=0$ and let $Q_{1}$ be the quarfs
with $q(0)=1$.  Let $q_{1},q_{2} \in Q_{0}$ or let $q_{1},q_{2} \in
Q_{1}$ then $\lambda = q_{2} - q_{1}$ satisfies (1) $\lambda(0)=0$ and
(2) $\lambda(v+w) = \lambda(v) + \lambda(w)$.  
Over the field $\bbZ_{2}$, the two conditions above imply that
$\lambda$ is a linear functional, \emph{i.e.}, $\lambda$ is in
$V^{*}$, the dual space of $V$.  Thus the number of distinct quadratic
forms is $2^{\dim V}$ for both $Q_{0}$ and $Q_{1}$.  Also since $b$ is
non-degenerate there exists a $t \in V$ such that
$\lambda(v) = b(v,t)$.

\begin{thm}
    \label{thm:unique}
    Let $q$ be a quarf. If $q = q \circ T_{w}$ for some $w\in V$, then 
    $w=0$.
\end{thm}
\begin{proof}
    $q(v) = (q\circ T_{w})(v)= q(v+w) = q(v) + q(w) + q(0) + b(v,w)$
    implying $b(v,w) = q(w) + q(0)$ for all $v\in V$.   The right
    hand side is independent of $v$ and the bilinear form is
    non-degenerate therefore $w=0$.
\end{proof}

\begin{thm}
    \label{thm:quarf}
    Let $q$ and $q'$ be quarfs; then there exists a $w \in V$ such that
    $q'_{w} = q'\circ T_{w} = q + \epsilon$ where $\epsilon = q'(w) + 
    q(0)$.
\end{thm}

\begin{proof}
    $q'(v) - q'(0)$ and $q(v) - q(0)$ are in $Q_{0}$ so their
    difference is in $V^{*}$ and given by $b(v,t)$ for some $t\in V$.
    Hence
    \begin{align*}
	q'_{w}(v) & = q'(v+w) \,, \\
	 & = q'(v) + q'(w) + q'(0) + b(v,w) \,, \\
	 & = (q'(0) + q(v) -q(0) + b(v,t)) + q'(w) + q'(0) + b(v,w) \,, \\
	 & = q(v) + b(v,t+w) + q'(w) + q(0)\,.
    \end{align*}
    Choose $w = t$ so that $q'_{w} = q + \epsilon$ where $\epsilon =
    q'(w)+q(0)$.
\end{proof}

In general $T_{z}$ does not act on $Q_{0}$ or $Q_{1}$ because $q(0)$
is not necessarily the same as $(q \circ T_{z})(0) = q(z)$.  This can
be easily fixed by defining $\mu_{z} : Q_{0} \to Q_{0}$ and $\nu_{z} :
Q_{1} \to Q_{1}$ by $(\mu_{z}q)(v) = (q\circ T_{z})(v) - q(z)$ and
$(\nu_{z}q)(v) = (q\circ T_{z})(v) - q(z) +1$.  The map $\kappa: q
\mapsto q-1$ gives an isomorphism of $Q_{0}$ with $Q_{1}$.

\begin{thm}
    \label{thm:symp-basis}
    There exists a symplectic basis $(\alpha_{j},\beta_{k})$ for the
    non-degenerate bilinear form $b$ such that $b(\alpha_{j},\beta_{k}) =
    \delta_{jk}$.  
\end{thm}
\begin{proof}
    (By induction) Assume the theorem for a vector space of dimension
    $2(g-1)$.  Chose $\alpha_{1}$ and find a $\beta_{1}$ such that
    $b(\alpha_{1},\beta_{1})=1$.  Complete $\alpha_{1}, \beta_{1}$ to
    a basis for the vector space $V$, $\dim V = 2g$.  With respect to
    this basis the matrix for $b$ is of the form
    \begin{equation*}
	b = \begin{pmatrix}
	    \begin{smallmatrix}
		0 & 1 \\
		1 & 0
	    \end{smallmatrix}
	     & S  \\
	    S^{t} & b'  
	\end{pmatrix}
    \end{equation*}
    Elementary row and column operations (done symmetrically and using
    the upper left block) can be used to define a new basis where
    $S=0$.  The problem is now reduced to finding a symplectic basis
    for $b'$, the $2(g-1)$ dimensional case.
\end{proof}

Choose a fixed symplectic basis $(\alpha_{j},\beta_{k})$ and define
the basic quarf $\qbasic$ by
\begin{equation}
    \qbasic(x\cdot 
    \alpha + y \cdot \beta) = b(x\cdot\alpha, y\cdot\beta) = x\cdot y\,.
    \label{eq:def-qbasic}
\end{equation}
It is easy to verify that $\qbasic$ is a quarf.  
\begin{thm}
    \label{thm:qbasic-zeroes}
    $\qbasic$ has $2^{g-1}(2^{g}
    +1)$ zeroes and takes the value $1$ at $2^{g-1}(2^{g} -1)$ points.
\end{thm}
\begin{proof}
    Easily proved by induction on $g$.  Note that $(x\cdot y)_{2g} =
    (x\cdot y)_{2(g-1)} + x_{g}\cdot y_{g}$.  If $(x\cdot y)_{2g}=0$
    then either $(x\cdot y)_{2(g-1)} =0$ and $x_{g}\cdot y_{g}=0$; or
    $(x\cdot y)_{2(g-1)} =1$ and $x_{g}\cdot y_{g}=1$.  Therefore the
    total number of zeroes is $2^{g-2}(2^{g-1}+1) \cdot 3 +
    2^{g-2}(2^{g-1}-1) \cdot 1$ as required.
\end{proof}

In our case the $2g$ dimensional vector space $V =
H^{1}(\Sigma, \bbZ_{2})$.  If $\xi,\eta \in H^{1}(\Sigma, \bbZ_{2})$
let   $\xi \wedge \eta$ denote the cup product.  If $\gen$ is the
generator of $H^{2}(\Sigma,\bbZ_{2}) \approx \bbZ_{2}$ then we write
$\xi \wedge \eta = (\xi \cup \eta) \gen$.  Our symplectic basis
$(\alpha, \beta)$ for $H^{1}(\Sigma, \bbZ_{2})$ is the mod~$2$
reduction of the dual symplectic basis in \eqref{eq:symp-ab}, 
still denoted  by $(\alpha_{j}, \beta_{k})$.  The basic quarf
$\qbasic$ is given by
\begin{equation}
    \qbasic( x_{j}\alpha_{j} + y_{j} \beta_{j}) = x \cdot y =
    (x_{j}\alpha_{j}) \cup (y_{k} \beta_{k})
    \label{eq:def-q}
\end{equation}
The following two corollaries will connect quarfs to spin structures.
\begin{cor}
    \label{cor:every-quarf}
    Let $\qbasic$ be the basic quarf; then every quarf $q$ can be put
    into the form $q = \qbasic\circ T_{w} + q(w)$ for some $w\in V$.
    The quarf $q$ has $2^{g-1}(2^{g}+1)$ zeroes if and only if
    $q(w)=0$, and $2^{g-1}(2^{g}-1)$ zeroes if and only if $q(w)=1$.
\end{cor}
\begin{proof}
    The first part is an immediate consequence of
    Theorem~\ref{thm:quarf}.  The second part follows from
    Theorem~\ref{thm:qbasic-zeroes} and the observation that $\qbasic
    \circ T_{w}$ has the same number of zeroes as $\qbasic$.
\end{proof}
\begin{cor}
    \label{cor:prin-hom-space}
    The map $w \mapsto \qbasic \circ T_{w}$ is a bijection onto the set
    of quarfs with $2^{g-1}(2^{g}+1)$ zeroes.
\end{cor}
\begin{proof}
    Injectiveness is a consequence of Theorem~\ref{thm:unique} and
    surjectiveness is a consequence of the previous corollary.
\end{proof}
This corollary can be restated: The set of quarfs with
$2^{g-1}(2^{g}+1)$ zeroes is a principal homogenous space for $V$.

\subsection{Quarfs and Spin Structures}
\label{sec:quarf-spin}

In \cite{Atiyah:spin}, M.F.~Atiyah, gives a map from spin structures
(square roots of the canonical bundle $K$ of $\Sigma$) to quadratic
refinements (quarfs).  Note that such square roots are a principal
homogeneous space of $H^{1}(\Sigma, \bbZ_{2})$ which is isomorphic to
flat line bundles whose square is the trivial bundle.  Choose a square
root $\spin$ of $K$ then the associated quarf $q_{\spin}$ at $w$ is
the mod~$2$ index of $\dbar\otimes I_{\spin \otimes F_{w}}$ where
$F_{w}$ is the square root of the trivial bundle determined by $w \in
H^{1}(\Sigma,\bbZ_{2})$.  Atiyah shows that $q_{\spin}$ is a quarf:
\begin{equation*}
    q_{\spin}(w_{1} + w_{2}) -q_{\spin}(w_{1}) -q_{\spin}(w_{2}) +
    q_{\spin}(0) = w_{1}\cup w_{2} \,.
\end{equation*}
For even spin structures $q_{\spin}(0)$ takes the value $0$ and for
odd spin structures the value $1$.  Atiyah also shows that $q_{\spin}$ has
$2^{g-1}(2^{g}+1)$ zeroes.
\begin{thm}
    \label{rem:quarf-spin}
    There is a unique even spin structure $\spinbasic$ such that 
    $q_{\spinbasic} = \qbasic$.
\end{thm}
\begin{proof}
    Let $\spin'$ be any spin structure.
    Corollary~\ref{cor:prin-hom-space} implies that there exists a
    unique $w$ with $q_{\spin'} = \qbasic \circ T_{w}$, and that
    $\qbasic = q_{\spinbasic}$ where $\spinbasic = \spin' \otimes F_{w}$ is an
    even spin structure.
\end{proof}

The element $w \in H^{1}(\Sigma,\bbZ/2\bbZ)$ determines a spin
structure $\spinbasic \otimes F_{w}$.  The theta characteristic $\delta \in
\thalf \bbZ/\bbZ$ associated with this spin structure is $\delta =
\thalf w$.  The holomorphic section of the determinant line bundle of
the $\dbar_{\spinbasic \otimes F_{w}}$ operator is proportional to
$\vartheta[\delta](0;\Omega)$.

\subsection{Quarfs and Modular Transformations}
\label{sec:quarf-symp-action}

In this section we describe how the symplectic group acts on spin
structures.  Fix a symplectic basis for $H^{1}(\Sigma,\bbZ)$ and which
gives the basic quarf on $\Sigma$.  Consider the set of all
spin structures over $\Met(\Sigma)$ and choose a metric $\bm{g} \in
\Met(\Sigma)$.  We have shown that there is a unique spin structure
$\spinbasic$ over $\bm{g}$ with the property that the associated quarf
$q_{\spinbasic}$ is the basic quarf $\qbasic$.  The map $f\in
\Diff(\Sigma)$ acts on $H_{1}(\Sigma,\bbZ)$ as a symplectic
transformation in $\Sp(2g, \bbZ)$ and therefore induces an action on
$H^{1}(\Sigma,\bbZ_{2})$ via $\Sp(2g,\bbZ_{2})$.  Let $\bm{g}'$ be the
transformed metric $f^{*}\bm{g}$.  There is a unique spin structure
$\spinbasic'$ over the $\bm{g}'$ such that the associated quarf
$q_{\spinbasic'}$ is the basic quarf $\qbasic$.  In general\footnote{
$f^{*}\spinbasic=\spinbasic'$ if $f \in \Torellibig(\Sigma)$; see also 
Section~\ref{sec:gam-2-1}.} $\spinbasic' \neq
f^{*}\spinbasic$.  To compare spin structures at different metrics we use
the basic quarf to single out the reference spin structures.  This
argument tells us that $f^{*}\spinbasic$ and $\spinbasic'$ differ by a square
root of the trivial bundle $F_{t}$ characterized by $t \in
H^{1}(\Sigma,\bbZ_{2})$:
\begin{equation}
    f^{*}\spinbasic = \spinbasic' \otimes F_{t}\,. 
    \label{eq:spin-transform}
\end{equation}
The action on square roots of the trivial bundle is $f^{*}F_{w} =
F_{f^{*}w}$ for $w \in H^{1}(\Sigma,\bbZ_{2})$.  Thus 
$f^{*}(\spinbasic \otimes F_{w}) = \spinbasic' \otimes \bigl( F_{t} \otimes
F_{f^{*}w} \bigr)$.  Below we express this equation explicitly when 
$f^{*}$ is represented by a matrix $\Lambda$.

\begin{thm}
    \label{thm:w-transf}
    If $\Lambda = \bigl( \begin{smallmatrix} A & B \\ C & D
    \end{smallmatrix} \bigr) \in \Sp(2g, \bbZ_{2})$, then $\Lambda$
    acts on spin structures and on quarfs such that $q_{\spinbasic' \otimes
    F_{w'}}(z) = q_{\spinbasic\otimes F_{w}}(\Lambda^{-1}z)$ where $w$ and
    $w'$ are related by
    \begin{equation}
	\begin{pmatrix}    
	    u' \\
	    v'
	\end{pmatrix}
	=
	\begin{pmatrix}    
	    D& C \\
	    B & A
	\end{pmatrix}
	\begin{pmatrix}
	    u \\
	    v
	\end{pmatrix}
	+
	\begin{pmatrix}
	    (CD^{t})_{d} \\
	    (AB^{t})_{d}
	\end{pmatrix},
	\label{eq:w-transf}
    \end{equation}
    and $w = \bigl(\begin{smallmatrix} u \\ v \end{smallmatrix}
    \bigr)$ in terms of the symplectic basis.  Also $u'\cdot w' =
    u\cdot w$ and therefore $\Sp(2g,\bbZ)$ maps even (odd) spin
    structure to even (odd) spin structures respectively.
\end{thm}

\begin{proof}
Under $f \in \Diff(\Sigma)$, the spin structure $\spinbasic' \otimes F_{w'}$
goes to a spin structure $\spinbasic\otimes F_{w}$ and the action on
$H^{1}(\Sigma,\bbZ_{2})$ is given by the symplectic transformation
$\Lambda \in \Sp(2g,\bbZ_{2})$.  Under the action of $f$ we have
\begin{equation}
    q_{\spinbasic' \otimes F_{w'}}(z) = q_{\spinbasic\otimes F_{w}}(\Lambda^{-1}
    z)
    \label{eq:quarf-transf-1}
\end{equation}
as a consequence of Corollary~\ref{cor:prin-hom-space}.  Letting
$z=0$ we see that the even (odd) spin structure $\spinbasic \otimes F_{w'}$
maps to an even (odd) spin structure $\spinbasic \otimes F_{w}$.  The
displayed equation above is equivalent to $\qbasic(z+w') =
\qbasic(\Lambda^{-1} z + w)$ which gives relations between $w$, $w'$
and $\Lambda$.  The definition of quarfs gives
\begin{equation*}
    \qbasic(z) + \qbasic(w') + b(z,w') = \qbasic(\Lambda^{-1}z) + 
    \qbasic(w) + b(\Lambda^{-1}z,w)\,;
\end{equation*}
hence
\begin{equation}
    \qbasic(z) = \qbasic(\Lambda^{-1}z) + b(\Lambda^{-1}z,w) - 
    b(z,w')\,.
    \label{eq:quarf-1}
\end{equation}
Replacing the argument $z$ above by $z_{1} + z_{2}$ and using the
defining properties of quarfs gives $b(z_{1},z_{2}) =
b(\Lambda^{-1}z_{1},\Lambda^{-1}z_{2})$.  If $\Lambda^{*}$ is the
adjoint with respect to the bilinear form $b$ then $\Lambda^{*}\Lambda
=I$.  The equation above may be rewritten as
\begin{equation}
    \qbasic(z) = \qbasic(\Lambda^{-1}z) + b(z,\Lambda w -w') \,.
    \label{eq:quarf-2}
\end{equation}
If we write $z = \bigl(\begin{smallmatrix} x \\ y \end{smallmatrix} 
\bigr)$ then 
\begin{align*}
    \qbasic(z) = x\cdot y &= (A^{t}x + C^{t}y) \cdot (B^{t}x + D^{t}y) + 
    b(z,\Lambda w - w') \\
    & = x\cdot y + x \cdot (AB^{t})_{d} + y \cdot (CD^{t})_{d} 
    + b(z,\Lambda w - w')\,.
\end{align*}
Since $w = \bigl(\begin{smallmatrix} u \\ v
\end{smallmatrix} \bigr)$,
\begin{equation*}
    \begin{pmatrix}    
	u' \\
	v'
    \end{pmatrix}
    =
    \begin{pmatrix}    
	D& C \\
	B & A
    \end{pmatrix}
    \begin{pmatrix}
	u \\
	v
    \end{pmatrix}
    +
    \begin{pmatrix}
	(CD^{t})_{d} \\
	(AB^{t})_{d}
    \end{pmatrix}.
\end{equation*}

Showing that $u'\cdot w' = u\cdot w$ (equivalent to the statement that
odd (even) is mapped to odd (even) spin structures) is a simple
algebraic computation requiring identities \eqref{eq:null} and
\eqref{eq:magic} whose proof we leave as an exercise.
\end{proof}
\begin{equation}
    0= (AB^{t})_{d} \cdot (CD^{t})_{d} =
    (A^{t}C)_{d} \cdot (B^{t}D)_{d}\,.
    \label{eq:null}
\end{equation}
\begin{equation}
    \begin{pmatrix}    
	(A^{t}C)_{d}\\
	(B^{t}D)_{d}
    \end{pmatrix}
    =
    \begin{pmatrix}    
	A^{t}& C^{t} \\
	B^{t} & D^{t}
    \end{pmatrix}
    \begin{pmatrix}    
	(DC^{t})_{d}\\
	(BA^{t})_{d}
    \end{pmatrix}
    \label{eq:magic}
\end{equation}

\subsection{The groups $\Gamma_{2} \subset \Gamma_{1,2} \subset \Sp(2g,\bbZ)$}
\label{sec:gam-2-1}

Let $\Sp(2g,\bbZ/2\bbZ)$ be the mod~$2$ reduction of $\Gamma_{1}=
\Sp(2g,\bbZ)$.  We have the exact sequence
\begin{equation*}
    1 \to \Gamma_{2} \to \Sp(2g,\bbZ) \xrightarrow{\bmod 2\;} 
    \Sp(2g,\bbZ/2\bbZ) \to 1
\end{equation*}
where $\Gamma_{2}$ is the normal subgroup of $\Gamma_{1}$ given by
\begin{equation}
    \Gamma_{2} = \left\{ \Lambda \in \Gamma_{1} \;\mid\; \Lambda = 
    I_{2g} \bmod 2 \right\}\,.
    \label{eq:Gamma-2}
\end{equation}
Let $\Gamma_{1,2} = \left\{ \Lambda \in \Gamma_{1} \;\mid\;
\qbasic(\Lambda z) = \qbasic(z) \bmod 2 \right\}$ where $\qbasic$ is
the basic quarf. Then $\Gamma_{2} \subset \Gamma_{1,2} \subset
\Gamma_{1}$, and
\begin{equation}
    \Gamma_{1,2} = \left\{ \Lambda \in \Gamma_{1} \;\mid\; 
    (AB^{t})_{d} = (CD^{t})_{d} = 0  \bmod 2 \right\}\,.
    \label{eq:Gamma-2-1-2}
\end{equation}
By \eqref{eq:magic}, the conditions defining $\Gamma_{1,2}$ imply that
$(A^{t}C)_{d} = (B^{t}D)_{d} = 0 \bmod 2$.  Note that if $f \in
\Diff(\Sigma)$ represents $\Lambda \in \Gamma_{1,2}$, then
$f^{*}\spinbasic=\spinbasic'$, \emph{i.e.}, $F_{t}$ is the trivial
bundle. See equations \eqref{eq:spin-transform} and \eqref{eq:w-transf}. 
It follows that the action of $f$ on a quarf is a linear 
transformation because the translation part of \eqref{eq:w-transf} 
vanishes.

\section{The Spinor Determinant}
\label{sec:spinor-det}

In our invariant the term $(\det'\dbar_{\delta}/\Nb_{\delta}^{2})^{n}$
appears where $\delta$ is an odd spin structure and $\dim M =2n$.  We
show that this term is related to the determinant section of the
preferred spin structure and $\vartheta(0;\Omega)$; see 
\eqref{eq:detp-section}.  A discussion that
applies to even spin structures may be found in
\cite[Section~5]{Alvarez-Gaume:1986es}.

The idea is basically the following.  Let $\Jacgo$ be the jacobian
torus associated with line bundles with first Chern class $\cclass_{1}
= g-1$.  The $\dbar$ operator coupled to those line bundles has index
zero and therefore there is a determinant line bundle with a canonical
section, a Quillen metric, \emph{etc.}.  The curvature of the Quillen
metric is the standard translationally invariant polarization form on
the torus.  This means that the determinant line bundle has a unique
holomorphic section.  The divisor for the determinant line bundle will
be a translate of the $\Theta$ divisor on the jacobian.  To analyze
this problem we choose a fiducial spin structure $\spinbasic$, and
think of $\Jacgo$ as $\{ \spinbasic \otimes F \;\mid\; F \in \Jacz\}$.
Equivalently a flat line bundle is equivalent to a flat connection and
we can study the holomorphic family of operators $D(A)=
\dbar_{\spinbasic} + \Azo$ where $\Azo \in \Hzo$ is a connection
representing the flat line bundle $F$.  The divisor of the section
$\det D(A)$ will be the standard $\Theta$ divisor
\cite{Alvarez-Gaume:1986es}.

How do we define $\det D(A)$?  This is discussed in detail in
\cite{Alvarez-Gaume:1986es} and we provide a brief overview motivated
by the path integral formulation of quantum field theory.  Assume we
have a chiral right moving Weyl spinor and a chiral left moving Weyl
spinor coupled to our flat connection $A$.  We know that this path
integral can be regulated in a gauge invariant way yielding the answer
$\det D(A)^{*}D(A)$.  If we only had a chiral fermion coupled to the
vector potential then the answer should be $\det D(A)$ (ill defined
for the moment).  We expect that since we can put the right and left
moving systems together into the path integral that $\det(D(A)^{*}D(A)
= (\det D(A)^{*})(\det D(A))$.  The work of Quillen
\cite{Quillen:det-line-bundle} shows that there is no holomorphic
factorization and thus $\det D(A)$ cannot be a holomorphic function of
$A \in \Hzo$ if the determinant multiplication property is to be
valid.  In fact in \cite[eq.~(5.13)]{Alvarez-Gaume:1986es} it is shown
that on the cover $\Hzo$ of the jacobian $\Jacz$, the section $\det
D(z)$ is given by the theta function with general characteristic
$\bigl[ \begin{smallmatrix} u \\
v \end{smallmatrix} \bigr]$ where $z= \Omega u + v$:
\begin{equation}
	\det D(z) = \left(\det \dbar(\spinbasic) \right)
	\frac{\vartheta[z](0)}{\vartheta(0)} = \left(\det
	\dbar(\spinbasic) \right) \frac{e^{i\pi u\cdot \Omega \cdot u +
	2\pi i u\cdot v}\; \vartheta(\Omega u + v)}{\vartheta(0)}\;.
    \label{eq:det-spinor}
\end{equation}
In the above $\det D(0) = \det \dbar(\spinbasic)$.  Because of the
exponential factor the last expression above, $\det D(z)$ is not a
holomorphic function of $z=\Omega u + v$.

\subsection{Bosonization Theorem for Odd Spin Structures}
\label{sec:bosonization}

If an odd spin structure has semi-characteristic $\delta = \Omega a +
b$, then we know that $\det D(\delta)=0$.  Consider $\det D(\delta +
z)$.  Note that
\begin{equation}
    \vartheta(\delta + z) = e^{-i\pi a\cdot \Omega \cdot a -2\pi i a 
    \cdot(z+b)}\; \vartheta[\delta](z)\,,
    \label{eq:theta-char}
\end{equation}
and therefore
\begin{equation}
	\left.  \frac{\partial\det D(\delta + z)}{\partial z^{j}}
	\right\rvert_{z=0} = \frac{\det \dbar(\spinbasic)}{\vartheta(0)}
	\; \frac{\partial \vartheta[\delta](0)}{\partial z^{j}}\,.
    \label{eq:det-spin-1}
\end{equation}
In the above we note that since $\delta$ is an odd spin structure $4
a\cdot b = 1 \bmod 2$.

To connect with $\det'\dbar_{\delta}/\Nb_{\delta}^{2}$ consider $(\det
D^{*}D)(\delta + z)$ in perturbation theory near $z=0$.  Let
$\lambda(z)$ be the smallest eigenvalue.  Since $\lambda(0)=0$,
\begin{equation}
    \lim_{z\to 0} \frac{(\det D^{*}D)(\delta + z)}{\lambda(z)}= 
    \det\nolimits'\dbar_{\delta}^{*}\dbar_{\delta}
    \text{ by definition.}
    \label{eq:det-pert}
\end{equation} 
The best way to compute the left hand side above is to change
viewpoint temporarily and think of $D(z+\delta)$ as a translate of
$\dbar_{\delta}$ by a flat connection.  Thus we define $\cD(z) =
\dbar_{\delta} + \Azo$ acting on sections of the line bundle
$\spin_{\delta}$ associated with the odd semi-characteristic
$\delta$.  Let $\phi(z)$ be the eigensection of $\cD^{*}\cD(z)$ with
eigenvalue $\lambda(z)$.  Note that $\phi(0) = h_{\delta}$ which a
the holomorphic section of $\dbar_{\delta}$.  We now do perturbation
theory.   Begin with the equation
\begin{equation}
    \cD^{*}(z)\cD(z) \phi(z) = \lambda(z) \phi(z)
    \label{eq:ev-eq}
\end{equation}
and differentiate:
\begin{align}
   \cD^{*} \frac{\partial \cD}{\partial z^{i}} \phi + \cD^{*} \cD
   \frac{\partial \phi}{\partial z^{i}}& = \frac{\partial
   \lambda}{\partial z^{i}} \phi + \lambda \frac{\partial
   \phi}{\partial z^{i}} \,,
   \label{eq:pert-1} \\
    \frac{\partial \cD^{*}}{\partial \zbar^{j}} \cD \phi +
    \cD^{*} \cD \frac{\partial \phi}{\partial \zbar^{j}}& =
    \frac{\partial \lambda}{\partial \zbar^{j}}\phi + \lambda 
    \frac{\partial \phi}{\partial \zbar^{j}}\,.
    \label{eq:pert-1b}
\end{align}
Evaluating the above at $z=0$ gives
\begin{align}
   \dbar_{\delta}^{*} \frac{\partial \cD(0)}{\partial z^{i}} h_{\delta} + 
   \dbar_{\delta}^{*}\dbar_{\delta}
   \frac{\partial \phi(0)}{\partial z^{i}}& = \frac{\partial
   \lambda(0)}{\partial z^{i}} h_{\delta}  \,,
   \label{eq:pert-1-0} \\
   \dbar_{\delta}^{*}\dbar_{\delta}
   \frac{\partial \phi(0)}{\partial \zbar^{j}}& =
    \frac{\partial \lambda(0)}{\partial \zbar^{j}} h_{\delta} \,.
    \label{eq:pert-1b-0}
\end{align}
Taking the inner product with $h_{\delta}$  gives
\begin{equation}
   0 = \frac{\partial \lambda(0)}{\partial z^{i}} = \frac{\partial
   \lambda(0)}{\partial \zbar^{j}}\,.
    \label{eq:lambda-der-1}
\end{equation}
Insert this into the previous equations to obtain
\begin{align}
       \left.  \frac{\partial}{\partial z^{i}} (\cD \phi)
       \right\rvert_{z=0} = \frac{\partial \cD(0)}{\partial z^{i}}
       h_{\delta} + \dbar_{\delta} \frac{\partial \phi(0)}{\partial
       z^{i}} & \in \ker \dbar_{\delta}^{*} \,, 
       \label{eq:pert-1-1} \\
   \frac{\partial \phi(0)}{\partial \zbar^{j}}& \in \ker 
   \dbar_{\delta}\,.
    \label{eq:pert-1b-1}
\end{align}

To show that $\partial^{2} \lambda(0)/\partial z^{i} \partial
z^{j}=0$, differentiate \eqref{eq:pert-1} one more time
and obtain
\begin{equation*}
   \cD^{*} \frac{\partial \cD}{\partial z^{i}}
   \frac{\partial\phi}{\partial z^{j}} + \cD^{*} \frac{\partial
   \cD}{\partial z^{j}} \frac{\partial \phi}{\partial z^{i}} +\cD^{*}
   \cD \frac{\partial^{2} \phi}{\partial z^{j} \partial z^{i}} =
   \frac{\partial^{2} \lambda}{\partial z^{j}\partial z^{i}} \phi
   +\frac{\partial \lambda}{\partial z^{i}}
   \frac{\partial\phi}{\partial z^{j}} + 
   \frac{\partial\lambda}{\partial z^{j}} \frac{\partial
   \phi}{\partial z^{i}} + \lambda \frac{\partial^{2}
   \phi}{\partial z^{j} \partial z^{i}}\,.
\end{equation*}
Evaluate this expression at $z=0$ and take the inner product with
$h_{\delta}$ to get the desired result.  Note that since the
eigenvalue $\lambda(z)$ is real, $\partial^{2}
\lambda(0)/\partial \zbar^{i} \partial \zbar^{j}=0$,
easily corroborated by differentiating \eqref{eq:pert-1b}.

Finally we consider the term $\partial^{2} \lambda(0)/\partial z^{i}
\partial \zbar^{j}$.  Differentiate \eqref{eq:pert-1} with respect to
$\zbar^{j}$ and obtain
\begin{gather}
    \frac{\partial \cD^{*}}{\partial \zbar^{j}} \frac{\partial
    \cD}{\partial z^{i}} \phi +\cD^{*} \frac{\partial \cD}{\partial
    z^{i}} \frac{\partial\phi}{\partial \zbar^{j}} + \frac{\partial
    \cD^{*}}{\partial \zbar^{j}} \cD \frac{\partial \phi}{\partial
    z^{i}} + \cD^{*} \cD \frac{\partial^{2} \phi}{\partial \zbar^{j}
    \partial z^{i}} 
    \nonumber \\
    = \frac{\partial^{2} \lambda}{\partial \zbar^{j}
    \partial z^{i}} \phi +\frac{\partial \lambda}{\partial z^{i}}
    \frac{\partial\phi}{\partial \zbar^{j}} + 
    \frac{\partial\lambda}{\partial \zbar^{j}} \frac{\partial
    \phi}{\partial z^{i}} + \lambda \frac{\partial^{2}
    \phi}{\partial \zbar^{j} \partial z^{i}}\,. 
    \label{eq:pert-2} 
\end{gather}
Evaluating at $z=0$ and regrouping terms gives
\begin{equation}
    \frac{\partial \cD^{*}(0)}{\partial \zbar^{j}} \left(
    \frac{\partial \cD(0)}{\partial z^{i}} h_{\delta} + \dbar_{\delta}
    \frac{\partial \phi(0)}{\partial z^{i}} \right) +
    \dbar_{\delta}^{*} \frac{\partial \cD(0)}{\partial z^{i}}
    \frac{\partial\phi(0)}{\partial \zbar^{j}} +
    \dbar_{\delta}^{*}\dbar_{\delta} \frac{\partial^{2}
    \phi(0)}{\partial \zbar^{j} \partial z^{i}} = \frac{\partial^{2}
    \lambda(0)}{\partial \zbar^{j} \partial z^{i}} h_{\delta} \,.
    \label{eq:pert2-0}
\end{equation}
Taking the inner product with $h_{\delta}$ gives
\begin{equation}
    \left\langle h_{\delta}, \frac{\partial \cD^{*}(0)}{\partial
    \zbar^{j}} \left( \frac{\partial \cD(0)}{\partial z^{i}}
    h_{\delta} + \dbar_{\delta} \frac{\partial \phi(0)}{\partial
    z^{i}} \right) \right\rangle = \frac{\partial^{2}
    \lambda(0)}{\partial \zbar^{j} \partial z^{i}} \langle h_{\delta},
    h_{\delta} \rangle \,.
    \label{eq:pert2-1}
\end{equation}
Let $\psi_{\delta} \in \ker \dbar_{\delta}^{*}$ then \eqref{eq:pert-1-1} tells
us that there exist constants $r_{i} \in \bbC$ such that
\begin{equation}
     \frac{\partial \cD(0)}{\partial z^{i}}
	h_{\delta} + \dbar_{\delta} \frac{\partial \phi(0)}{\partial
	z^{i}}  = r_{i}\; \psi_{\delta}\,.
    \label{eq:def-ri}
\end{equation}
Taking the inner product with $\psi_{\delta}$ and using 
$\dbar_{\delta}^{*}\psi_{\delta}=0$,
\begin{equation}
    \left\langle \psi_{\delta}, \frac{\partial \cD(0)}{\partial z^{i}} 
    h_{\delta} \right\rangle = r_{i}\; \langle \psi_{\delta}, \psi_{\delta} \rangle\,.
    \label{eq:soln-ri}
\end{equation}
Next we manipulate \eqref{eq:pert2-1}:
\begin{align*}
    \frac{\partial^{2} \lambda(0)}{\partial \zbar^{j} \partial z^{i}}
    \langle h_{\delta}, h_{\delta} \rangle & = \left\langle
    h_{\delta}, \frac{\partial \cD^{*}(0)}{\partial \zbar^{j}} \left(
    \frac{\partial \cD(0)}{\partial z^{i}} h_{\delta} + \dbar_{\delta}
    \frac{\partial \phi(0)}{\partial z^{i}} \right) \right\rangle, \\
    & = \left\langle \frac{\partial \cD(0)}{\partial z^{j}}
    h_{\delta}, \left( \frac{\partial \cD(0)}{\partial z^{i}}
    h_{\delta} + \dbar_{\delta} \frac{\partial \phi(0)}{\partial
    z^{i}} \right) \right\rangle, \\
    & = \left\langle r_{j}\psi_{\delta} - \dbar_{\delta} \frac{\partial 
    \phi(0)}{\partial z^{j}}, r_{i} \psi_{\delta}  \right\rangle, \\
    & = \overline{r_{j}} \; r_{i} \;\langle \psi_{\delta}, \psi_{\delta} \rangle\,,
\end{align*}
where we used $\dbar_{\delta}^{*} \psi_{\delta}=0$. We rewrite this as
\begin{equation}
    \frac{\partial^{2} \lambda(0)}{\partial \zbar^{j} \partial z^{i}}
    = \frac{1}{\langle h_{\delta}, h_{\delta} \rangle \langle
    \psi_{\delta}, \psi_{\delta} \rangle} \; \overline{\left\langle
    \psi_{\delta}, \frac{\partial \cD(0)}{\partial z^{j}} h_{\delta}
    \right\rangle } \; \left\langle \psi_{\delta}, \frac{\partial
    \cD(0)}{\partial z^{i}} h_{\delta} \right\rangle.
    \label{eq:ev-sec-order}
\end{equation}
Note that as required, the answer is independent of how we scale
$h_{\delta}$ and $\psi_{\delta}$ and that the second derivative
factorizes. We now derive a simple formula for 
\eqref{eq:ev-sec-order}.

Observe that $\dbar_{\delta}: \Lhz \to \Lho \approx 
\left(\Lhz\right)^{*}$ and that the hermitian inner product on our 
line bundles gives us a conjugate linear isomorphism between $\Lhz$ 
and $\Lho$. Hence we can identify $h_{\delta}$ with 
$\psi_{\delta}$ and 
\begin{align*}
    \left\langle \psi_{\delta}, \frac{\partial
	\cD(0)}{\partial z^{i}} h_{\delta} \right\rangle &= 
	\frac{1}{2i} \int_{\Sigma} \frac{\partial
	\cD(0)}{\partial z^{i}} \wedge h_{\delta}^{2}\,,\\
	&=  \frac{1}{2i} \int_{\Sigma} 2\pi i \,
	(\Omega-\bar{\Omega})^{-1}_{ij}\bar{\omega}_{j}\wedge 
	\sum_{k}\frac{\partial \vartheta[\delta](0)}{\partial z^{k}} 
	\omega_{k} \\
	& =  \pi \; \frac{\partial
	\vartheta[\delta](0)}{\partial z^{i}}
\end{align*}
Thus  \eqref{eq:ev-sec-order} becomes
\begin{equation}
    \frac{\partial^{2} \lambda(0)}{\partial \zbar^{j} \partial z^{i}}
    = \pi^{2} \; \frac{1}{\langle h_{\delta},
    h_{\delta} \rangle^{2}} \; \overline{\frac{\partial
    \vartheta[\delta](0)}{\partial z^{j}}} \; \frac{\partial
    \vartheta[\delta](0)}{\partial z^{i}}\,.
    \label{eq:ev-sec-order-1}
\end{equation}
\cite[eq.~(5.11)]{Alvarez-Gaume:1986es} states
\begin{equation}
    \det D(u)^{*}D(u) = \frac{\det
    \dbar_{\spinbasic}^{*}\dbar_{\spinbasic}}{\left\lvert\vartheta(0) 
    \right\rvert^{2}}\;
    \left\lvert\vartheta(u)\right\rvert^{2}
    e^{i\pi(u-\bar{u})(\Omega- \bar{\Omega})^{-1}(u-\bar{u})}\,.
    \label{eq:det-DD}
\end{equation}
Letting $u=\delta +z$ we find
\begin{equation}
    \det D(u)^{*}D(u) = \frac{\det
    \dbar_{\spinbasic}^{*}\dbar_{\spinbasic}}{\left\lvert\vartheta(0) 
    \right\rvert^{2}}\;
    \left\lvert e^{-i\pi a\cdot \Omega \cdot a -2\pi i a 
    \cdot(z+b)}\; \vartheta[\delta](z) \right\rvert^{2}
    e^{i\pi(u-\bar{u})\cdot(\Omega- \bar{\Omega})^{-1} \cdot(u-\bar{u})}\,.
    \label{eq:det-DD-1}
\end{equation}
The Taylor series of the above is
\begin{align}
    \det D(u)^{*}D(u) &= \frac{\det
    \dbar_{\spinbasic}^{*}\dbar_{\spinbasic}}{\left\lvert\vartheta(0)
    \right\rvert^{2}}\; \left\lvert e^{-i\pi a\cdot \Omega \cdot a}\;
    \sum_{j} z^{j} \; \frac{\partial \vartheta[\delta](0)}{\partial
    z^{j}} \right\rvert^{2} e^{i\pi a\cdot(\Omega-
    \bar{\Omega})\cdot a} + O(z^{3})\,,
    \nonumber \\
    &= \frac{\det
    \dbar_{\spinbasic}^{*}\dbar_{\spinbasic}}{\left\lvert\vartheta(0)
    \right\rvert^{2}}\; \left\lvert \sum_{j} z^{j} \; \frac{\partial
    \vartheta[\delta](0)}{\partial z^{j}} \right\rvert^{2} +
    O(z^{3})\,.
    \label{eq:det-DD-2}
\end{align}
Since for a normalized spinor $\hat{h}_{\delta}$ we defined
$h_{\delta} = N_{\delta} \hat{h}_{\delta}$ and therefore $\lvert
N_{\delta} \rvert^{2} = \langle h_{\delta},h_{\delta} \rangle$ we have
\begin{equation}
    \pi^{2} \; \frac{\det\nolimits'
    \dbar_{\delta}^{*}\dbar_{\delta}}{\lvert \langle h_{\delta},
    h_{\delta} \rangle \rvert^{2}} = \pi^{2} \; \frac{\det\nolimits'
    \dbar_{\delta}^{*}\dbar_{\delta}}{\lvert N_{\delta} \rvert^{4}} =
    \frac{\det
    \dbar_{\spinbasic}^{*}\dbar_{\spinbasic}}{\left\lvert\vartheta(0)
    \right\rvert^{2}}\,.
    \label{eq:det-prime}
\end{equation}
This formula says that if $\delta$ and $\epsilon$ are odd spin 
structures then
\begin{equation}
    \frac{\det\nolimits' \dbar_{\delta}^{*}\dbar_{\delta}}{\lvert
    \langle h_{\delta}, h_{\delta} \rangle \rvert^{2}}
    =\frac{\det\nolimits' \dbar_{\epsilon}^{*}\dbar_{\epsilon}}{\lvert
    \langle h_{\epsilon}, h_{\epsilon} \rangle \rvert^{2}}\,.
    \label{eq:det-univ}
\end{equation}
    In other words the function $\det'( \dbar_{\delta}^{*}
    \dbar_{\delta})/\lvert \langle h_{\delta}, h_{\delta} \rangle
    \rvert^{2}$ on $\Teichodd(\Sigma)$ descends to a function on
    $\Teich(\Sigma)$.

This leads to a new result about bosonization for odd spin structures.
The standard bosonization result \cite{AlvarezGaume:1987vm} is that for
every \emph{even} spin structure $\beta$
\begin{equation}
    \left( \frac{\vol \Sigma \; \det H}{\detz} 
    \right)^{1/2} = \frac{1}{2}\; \frac{\det \dbar^{*}_{\beta}\dbar_{\beta}}{\lvert 
    \vartheta[\beta](0) \rvert^{2}}\,.
    \label{eq:bosonization-even}
\end{equation}
If we compare this with \eqref{eq:det-prime} we see that
for every \emph{odd} spin structure $\delta$
\begin{equation}
    \left( \frac{\vol \Sigma \; \det H}{\detz} \right)^{1/2} =
    \frac{\pi^{2}}{2} \;     \frac{\det\nolimits^{\prime}
    \dbar^{*}_{\delta}\dbar_{\delta}}{\left\lvert \langle h_{\delta}, 
    h_{\delta} \rangle \right\rvert^{2}}\,.
    \label{eq:bosonization-odd}
\end{equation}

\subsection{Trivializing}
\label{sec:trivializing}

We discuss the trivialization of the determinant line bundle for the
Dirac operator over spin Teichmuller space $\Teichspin(\Sigma)$ and in
the process we clarify what we mean by $\det \dbar$ throughout much of
this article.  $\Teich(\Sigma)$ is a convex space with the topology of
$(\bbR \times \bbR_{+})^{3g-3}$ and therefore line bundles are
trivializable.
\begin{figure}[tbp]
    \centering
    \includegraphics[width=0.3\textwidth]{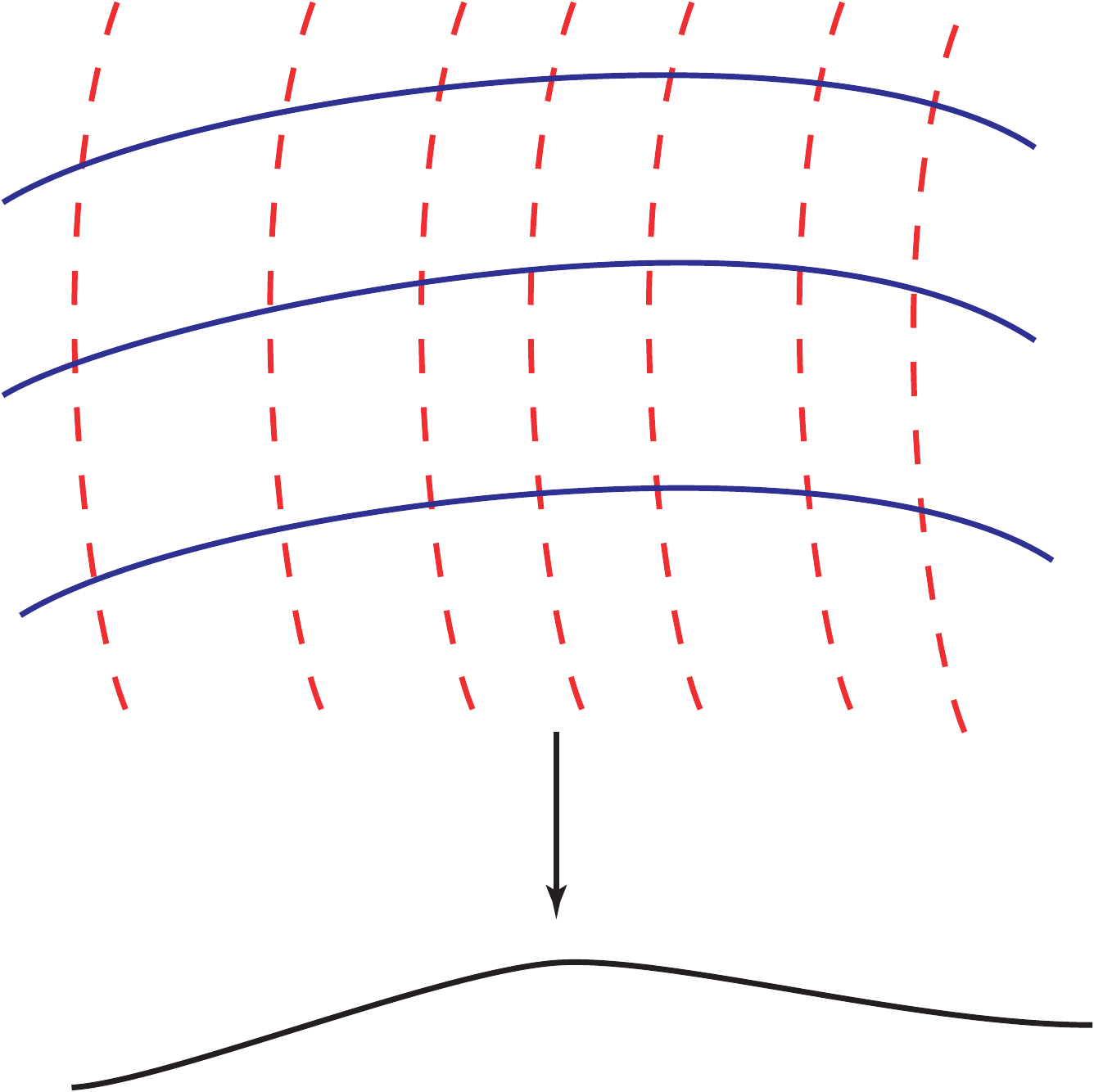}
    \begin{quote}
	\caption[xteich]{\small The above is a schematic figure of
	spin Teichmuller space, $\Teichspin(\Sigma)$, and how it fits
	inside of the bundle $\Teich(\Sigma) \times \Jacz$.  The base
	is $\Teich(\Sigma)$ and the fibers are isomorphic to $\Jacz$.
	The preferred spin structure $\spinbasic$ associated with the basic
	quarf is taken to be the origin of the jacobian and all the
	other spin structures correspond to the half period points of
	the torus.  Spin Teichmuller space is represented by the
	horizontal leaves.}
	\label{fig:teich}
    \end{quote}
\end{figure}
Spin Teichmuller space $\Teichspin(\Sigma)$ and spin moduli space
$\Modspin(\Sigma)$ are respectively finite covers of Teichmuller space
$\Teich(\Sigma)$ and moduli space $\Mod(\Sigma)$ with projection maps 
$\pi$ and $\pi_{1/2}$. They are related by the following commutative 
diagram:
\begin{equation}
    \begin{CD}
	\Teichspin(\Sigma) @>\pi_{1/2}>>   \Modspin(\Sigma) \\
	@V\text{forgetful}VV       @VV\text{forgetful}V  \\
	\Teich(\Sigma)   @>\pi>>     \Mod(\Sigma)
    \end{CD}
    \label{eq:comm-diag}
\end{equation}
The vertical arrows correspond to  forgetful functors that ignore 
the spin structure information.

Consider the conformal field theory of two chiral $(1/2,0)$
fermions\footnote{With an even number of chiral fermions we can frame
the discussion in terms of determinant line bundles and not pfaffian
line bundles.}.  According to Friedan and Shenker
\cite{Friedan:1986ua}, the expectation value of the energy momentum
tensor $\VEV{T_{zz}} (dz)^{2}$ of that conformal field theory gives a
flat hermitian connection $\mathcal{A}$ on a holomorphic line bundle
over spin moduli space $\Modspin(\Sigma)$.  The line bundle in
question $\cL \to \Modspin(\Sigma)$ is the determinant line bundle
associated with the Dirac operator and the ``partition function'' is
the determinant section.  We can pull back the line bundle and
connection to the cover $\Teichspin(\Sigma)$.

Restrict to the leaf $\Teichspin(\Sigma,\qbasic)$ in
$\Teichspin(\Sigma)$ associated with the basic quarf $\qbasic$.  On
this leaf we have the determinant line bundle $\Ltil(\qbasic) \to
\Teichspin(\Sigma,\qbasic)$ of the preferred spin structure
$\spinbasic$.  Next we trivialize this line bundle.  The pulled back
expectation value of the energy momentum tensor gives a flat
holomorphic connection $\widetilde{\mathcal{A}}$ on the determinant
line bundle $\Ltil(\qbasic) \to \Teichspin(\Sigma,\qbasic)$.  Pick a
distinguished point $\spinbasic_{0} \in \Teichspin(\Sigma,\qbasic)$
and a reference point $\hat{\sigma}_{0}$ on the fiber of
$\Ltil_{\spinbasic_{0}}$ over $\spinbasic_{0}$.  Once we have picked
$\hat{\sigma}_{0}$ we can identify the fiber $\Ltil_{\spinbasic_{0}}$
with $\bbC$.  The trivialization is along the discussion in
Section~\ref{sec:trivi-det}.  Let $\sigma$ be a point on the fiber
over $\spin \in \Teichspin(\Sigma,\qbasic)$ and parallel transport
$\sigma$ along any curve from $\spin$ to $\spinbasic_{0}$.  The choice
of curve is irrelevant because $\Teichspin(\Sigma,\qbasic)$ is
contractible and the connection is flat. The trivialization map 
$\phi: \Ltil(\qbasic) \to \Teichspin(\Sigma,\qbasic) \times \bbC$ is 
given by
\begin{equation}
    \sigma(\spin) = (\phi\sigma)(\spin)\; \exp\left( - 
    \int_{\spinbasic_{0}}^{\spin} \widetilde{\mathcal{A}} \right) 
    \hat{\sigma}_{0}\;.
    \label{eq:Teich-triv}
\end{equation}
It follows from the definition of spin Teichmuller space that 
$\Teichspin(\Sigma,\qbasic)$ is isomorphic to $\Teich(\Sigma)$ and we 
implicitly think of the trivialized line bundle as being $\Teich(\Sigma) 
\times \bbC$.

We trivialize $\Ltil \to \Teichspin(\Sigma)$ by embedding
$\Teichspin(\Sigma)$ in $\Teich(\Sigma) \times \Jacz$ as depicted in
Figure~\ref{fig:teich}.  Fix a point $\tau \in \Teich(\Sigma)$ and
consider the family of operators $\dbar(\spinbasic \otimes F)$ were $F
\in \Jacz$ is a flat line bundle and $\spinbasic$ is the spin
structure associated with the basic quarf $\qbasic$.  By going to the
cover $\Hzo$ of $\Jacz$ we can trivialize the family using a parallel
transport method analogous to the one discussed in
Section~\ref{sec:trivi-det}.  We can restrict to the half lattice
points in $\Hzo$ and we obtain that if $\delta = \half w$ where $w \in
H(\Sigma,\bbZ/2\bbZ)$ then
\begin{equation}
    \det \dbar(\spinbasic \otimes F_{w}) = 
    \frac{\vartheta[\delta](0;\Omega)}{\vartheta(0;\Omega)}\; 
    \det\dbar(\spinbasic)\,,
    \label{eq:spin-triv}
\end{equation}
according to \eqref{eq:det-spinor}.  For the moment
$\det\dbar(\spinbasic)$ is a point on the fiber of the determinant
line at $\spinbasic \in \Teichspin(\Sigma,\qbasic)$ over the point
$\tau \in \Teich(\Sigma)$.  Next we use the trivialization of 
$\Ltil(\qbasic) \to \Teichspin(\Sigma,\qbasic)$ previously described 
and this is how we interpret the determinant section.
Since we have the Quillen metric, $\det
\dbar(\spinbasic)$ is known up to a phase.

Our conformal field theory of chiral spinors is associated with the
determinant line bundle $\cL \to \Modspin(\Sigma)$ with holomorphic
partition section $Z$.  This holomorphic section satisfies the
parallel transport equation $\partial_{\Modspin(\Sigma)}Z+
\mathcal{A}Z=0$, and can be pulled back to a section $\widetilde{Z}$
of $\Ltil \to \Teichspin(\Sigma)$ where it satisfies
$\partial_{\Teichspin(\Sigma)}\widetilde{Z} + \widetilde{\mathcal{A}}
\widetilde{Z}=0$.  In the trivialization just described,
$\widetilde{Z}$ becomes a function that we denote by $\det\dbar$.
Note that after restriction to $\Teichspin(\Sigma,\qbasic)$, the
holomorphic determinant section $\det\dbar(\spinbasic)$ satisfies the
parallel transport equation $\partial_{\Teich(\Sigma)}\det\dbar+
\widetilde{\mathcal{A}}\, \det\dbar=0$.

Let $\mathscr{D} \to \Teichspin(\Sigma)$ be the trivialized
determinant line bundle as described in the previous paragraphs.  If
$f \in \Diff(\Sigma)$ represent $\Lambda \in \Sp(2g,\bbZ)$ then
$f^{*}\mathscr{D} = \mathscr{D} \otimes \mathscr{C}$ where
$\mathscr{C} \to \Teichspin(\Sigma)$ is a circle bundle because $f$
acts isometrically.  Consider the section $\det\dbar$ of
$\mathscr{D}$.  Since $f$ acts by isometries,
$f^{*}(\det\dbar)_{f^{*}\mathscr{D}} = u \cdot
(\det\dbar)_{\mathscr{D}}$ where $u$ is a section of $\mathscr{C}$ and
$\lvert u \rvert =1$.  Since the determinant sections vary
holomorphically over $\Teichspin(\Sigma)$, $u$ is a constant section.
The dependence of $u$ on $f$ is only through $\Lambda$ since we have
already taken into account the action of the normal subgroup
$\Diff_{0}(\Sigma)$.

\subsection{Defining $\detp\dbar$}
\label{sec:sec-detp}

Having trivialized the bundle, we can define $\detp\dbar$ by
exploiting the embedding of $\Teichspin(\Sigma)$ in $\Teich(\Sigma)
\times \Jacz$.  We rewrite \eqref{eq:det-spinor} as
\begin{align}
    \det D(\delta + z) &= \frac{\det \dbar(\spinbasic)}{\vartheta(0)} \; 
    e^{i\pi (a +u)\cdot \Omega \cdot (a+u) + 2\pi i (a+u) \cdot(b+v)}\; 
    \vartheta(\delta + z)\,, 
    \nonumber \\
    & = \frac{\det \dbar(\spinbasic)}{\vartheta(0)} \; 
    e^{i\pi (a +u)\cdot \Omega \cdot (a+u) + 2\pi i (a+u) \cdot(b+v)}
    \nonumber \\
    & \quad\times  e^{-i \pi a\cdot \Omega a - 2\pi i u\cdot (z+b)}\;
    \vartheta[\delta](z) \,
    \nonumber \\
    & = \frac{\det \dbar(\spinbasic)}{\vartheta(0)} \; e^{i \pi u\cdot
    \Omega u + 2\pi i u\cdot (v+b)}\; \vartheta[\delta](z) \,.
    \label{eq:det-spinor-1}
\end{align}
From \eqref{eq:ev-sec-order-1} we have
\begin{equation}
    \lambda(z) = \left\lvert 
    \frac{\pi}{N_{\delta}^{2}}\; \sum_{j=i}^{g} z^{j} \frac{\partial 
    \vartheta[\delta](0)}{\partial z^{j}} \right\rvert^{2} + 
    O(z^{3})\,,
    \label{eq:ev-square}
\end{equation}
and we can take an approximate ``holomorphic square root''
\begin{equation}
    \sqrt{\lambda(z)}= e^{i\psi}\;
    \frac{\pi}{N_{\delta}^{2}}\; \sum_{j=i}^{g} z^{j} \frac{\partial 
    \vartheta[\delta](0)}{\partial z^{j}}  + 
    O(z^{2})\,.
    \label{eq:ev-square-root}
\end{equation}
In the above we introduced an arbitrary phase $\psi$.  The subleading
terms will not have holomorphicity properties.  Thus we can define
$\det'\dbar_{\delta}$ by mimicking\footnote{This is what the path
integral suggests as a possible definition.} \eqref{eq:det-pert} and
defining
\begin{equation}
    \det\nolimits' \dbar_{\delta} = \lim_{z\to 0} \frac{D(\delta +
    z)}{\sqrt{\lambda(z)}}.
    \label{eq:def-detp-1}
\end{equation}
We  compute the above with the result
\begin{equation*}
    \frac{\detp \dbar_{\delta}}{N_{\delta}^{2}} =
    e^{-i\psi}\; \frac{1}{\pi}\; \frac{\det
    \dbar(\spinbasic)}{\vartheta(0)} \,.
\end{equation*}
The arbitrary phase in the normalization factor $N_{\delta}$ can be 
chosen so that $e^{i\psi}=1$ and
\begin{equation}
    \frac{\detp \dbar_{\delta}}{N_{\delta}^{2}} =
    \frac{1}{\pi}\; \frac{\det
    \dbar(\spinbasic)}{\vartheta(0)} \,.
    \label{eq:detp-section}
\end{equation}
This agrees with explicit results in genus one.

\section{Geometric Symplectic Action on Determinants}
\label{sec:geom-symp-det}

\subsection{Quillen Isometry}
\label{sec:warmup}

The determinant for the Dirac laplacian is given by 
\begin{equation}
    (\det \dbar^{*}\dbar)(\spinbasic \otimes F_{w}) = (\det \dbar^{*}\dbar)(\spinbasic) \;
    \left\lvert \frac{\vartheta[\delta](0; \Omega)}{\vartheta(0 ; 
    \Omega)} \right\rvert^{2}\,.
    \label{eq:det-spin}
\end{equation}
where $\spinbasic$ is the spin structure described by the basic quarf
$\qbasic$ and $\delta = \tfrac{1}{2} w \bmod \bbZ$, see
\eqref{eq:det-DD-1}.

Let $\Lambda \in \Sp(2g,\bbZ)$ and let $f \in \Diff(\Sigma)$ be a
representative for $\Lambda$. Because of the functorial properties of 
$\dbar$, the eigenvalues are preserved under the action of $f$ and we 
conclude that
\begin{equation}
    (\det \dbar^{*}\dbar)(\spinbasic' \otimes F_{t} \otimes F_{f^{*}w}) =
    (\det \dbar^{*}\dbar)(\spinbasic \otimes F_{w})\,.
    \label{eq:quillen-isom}
\end{equation}
\begin{remark}
    \label{rem:isometry}
    Equation~\eqref{eq:quillen-isom} states that the diffeomorphism
    $f$ acts as an isometry with respect to the Quillen metric on the
    $f$-related determinant line bundles.
\end{remark}
\begin{remark}\label{rem:isom-gen}
    The Quillen isometry is more general: if $F_{w}$ is replaced by an
    arbitrary flat line bundle $F_{\chi}$ corresponding to character
    $\chi$, then \eqref{eq:quillen-isom} is still valid.
\end{remark}

Using the notation that $F_{w'} = F_{f^{*}w} \otimes F_{t}$, $\delta'
= \thalf w'$, \emph{etc.}, a straightforward computation gives
\begin{align*}
    (\det \dbar^{*}\dbar)(\spinbasic' \otimes F_{w'}) &= (\det 
    \dbar^{*}\dbar)(\spinbasic') \;
    \left\lvert \frac{\vartheta[\delta'](0; \Omega')}{\vartheta(0 ; 
    \Omega')} \right\rvert^{2}\,, \\
    & = (\det \dbar^{*}\dbar)(\spinbasic') \; \left\lvert
    \frac{\varepsilon(\Lambda) e^{-i\pi\phi(\delta)} \det(C \Omega +
    D)^{1/2} \; \vartheta[\delta](0; \Omega)}{ \varepsilon(\Lambda)
    \det(C\Omega + D)^{1/2} \; \vartheta(0 ; \Omega)}
    \right\rvert^{2}\,, \\
    & = \lvert\det(C\Omega+D)\rvert\; \lvert \vartheta[\delta](0;
    \Omega)\rvert^{2}\; \frac{(\det \dbar^{*}\dbar)(\spinbasic')}{
    \lvert \vartheta(0;\Omega')\rvert^{2}}\,.
\end{align*}
Hence
\begin{equation}
    \frac{(\det \dbar^{*}\dbar)(\spinbasic)}{\lvert
    \vartheta(0;\Omega')\rvert^{2}} = \lvert\det(C\Omega+D)\rvert\;
    \frac{(\det
    \dbar^{*}\dbar)(\spinbasic)}{\lvert\vartheta(0;\Omega)\rvert^{2}}
    \;.
    \label{eq:det-sq-1}
\end{equation}

\subsection{Determinant Section for Spin Bundles}
\label{sec:det-spin-bundle}

In this section we study in more detail the Quillen isometry. 
Remark~\ref{rem:isom-gen} tells us that in general we expect a 
relationship
\begin{equation*}
    (\det\dbar)(\spinbasic'\otimes F_{t} \otimes f^{*}F_{\chi}) = 
    e^{i\xi(\chi)}\, (\det\dbar)(\spinbasic \otimes F_{\chi}).
\end{equation*}
A difficulty with the above is $\spinbasic \otimes F_{\chi}$ is an odd
spin structure then the determinant sections vanish and it is not
clear that the phase can be defined.  In this section we develop a
limit process that uniquely determines the phase for an odd spin
structure.

Fix a symplectic basis for $H_{1}(\Sigma,\bbZ)$.  Let $F_{\chi}$ be
the flat line bundle specified by character $\chi$.  We use $F_{\chi}$
as a perturbation and assume it is close to the trivial bundle.
Identify the character $\chi$ with a point $\zeta$ near the origin of
$\bbC^{g}$.  The assignment of $\zeta$ is given by the standard
coordinates on $\Hzo$, the cover of $\Jacz$.  We abuse notation and
denote $F_{\chi}$ by $F_{\zeta}$.  Let $\spinbasic \otimes F_{w}$ be
an odd spin structure, $w\in H^{1}(\Sigma,\bbZ/2\bbZ)$.  Trivializing
the determinant line bundle over $\Hzo$ for the family of $\dbar$
operators acting on $\spinbasic \otimes F_{\zeta}$ gives
\begin{equation}
    (\det\dbar)( \spinbasic \otimes F_{w} \otimes F_{\zeta}) = 
    (\det\dbar)(\spinbasic)\; 
    \frac{\vartheta[\delta + \zeta](0; \Omega)}{\vartheta(0 ; 
    \Omega)}\,.
    \label{eq:det-spinor-4}
\end{equation}
using\footnote{The algebraic geometry $\vartheta$-function conventions
use semi-characteristics $\delta = \thalf w \in \thalf \bbZ/\bbZ$.  At
times we abuse notation using $w$ and $\delta$ interchangeably.
Shifting $\delta$, $\zeta$ by a lattice vector may lead to phases in
the formula above.} \eqref{eq:det-spin-1} and \eqref{eq:theta-char}.

Let $f \in \Diff(\Sigma)$ that represents $\Lambda \in 
\Sp(2g,\bbZ)$. Since $f$ acts as an isometry on the determinant line 
bundles, see Section~\ref{sec:warmup}, we obtain
\begin{equation}
    (\det \dbar)(\spinbasic' \otimes F_{t} \otimes F_{f^{*}w} \otimes
    F_{f^{*}z}) = e^{i\xi(w+z;\Lambda)}\; (\det \dbar)(\spinbasic \otimes
    F_{w} \otimes F_{z})\,,
    \label{eq:det-transf-iso}
\end{equation}
where $\xi$ is a phase and $F_{t}$ is the ``translational shift'' 
flat line bundle discussed previously. 

Observe that
\begin{align}
    e^{i\xi(0,\Lambda)}\; 
    \frac{(\det\dbar)(\spinbasic)}{\vartheta(0;\Omega)}
    & = \frac{(\det \dbar)(\spinbasic' \otimes F_{t})}{\vartheta(0;\Omega)} 
    \nonumber \\
     & = \frac{\vartheta[t](0;\Omega')}{\vartheta(0;\Omega)}\; 
     \frac{(\det\dbar)(\spinbasic')}{\vartheta(0;\Omega')}\; 
     \nonumber \\
     &= \varepsilon(\Lambda) e^{-i\pi\phi(0,\Lambda)}
     \det(C\Omega+D)^{1/2}\;
     \frac{(\det\dbar)(\spinbasic')}{\vartheta(0;\Omega')}\,,
     \nonumber \\
     &= \varepsilon(\Lambda) 
     \det(C\Omega+D)^{1/2}\;
     \frac{(\det\dbar)(\spinbasic')}{\vartheta(0;\Omega')}\,,
     \label{eq:det-t-1}
\end{align}
where $\phi(0,\Lambda)=0$ by \eqref{eq:symp-phi}. 

\begin{remark}
    In genus one, fermionic determinants are of the form
    $\vartheta[\delta](\tau)/\eta(\tau)$, therefore $\det/\vartheta =
    1/\eta$ is universal for all even spin structures.  This is a
    special case of $(\det \dbar)(\spinbasic \otimes
    F_{w})/\vartheta[\delta](0) = (\det
    \dbar)(\spinbasic)/\vartheta(0)$.  Under the modular
    transformation $T(\tau)= \tau' = \tau+1$ we have that
    $(\det\dbar)(\spinbasic' \otimes F_{t}) =
    \vartheta_{01}(0;\tau')/\eta(\tau')$ and $(\det\dbar)(\spinbasic)
    = \vartheta_{00}(0;\tau)/\eta(\tau)$.  Note that
    $\vartheta_{01}(0;\tau')/\eta(\tau') = e^{-2\pi i /24}\,
    \vartheta_{00}(0;\tau)/\eta(\tau)$ and so we conclude that
    $e^{i\xi(0,T)}=e^{-2\pi i/24}$ in this example.
\end{remark}

\begin{remark}
    According to \eqref{eq:det-t-1}, $\vartheta(0;\Omega)/
    \det\dbar(\spinbasic)$ is a modular form of weight $1/2$.  It is
    constant along the fibers in the covering $\Teichspin(\Sigma) \to
    \Teich(\Sigma)$ in the sense that for even spin structures we use
    result \eqref{eq:spin-triv} and for odd spin structures we use
    \eqref{eq:detp-section}.  For this reason we can think of
    $\vartheta(0;\Omega)/ \det\dbar(\spinbasic)$ as defined on
    $\Teich(\Sigma)$ as a higher genus generalization of the
    Dedekind $\eta$-function.
\end{remark}

\begin{remark}
    Even in genus one, the computation of the modular transformations
    properties of the Dedekind $\eta$-function using the geometry of
    determinant line bundles is a very technical exercise
    \cite{Atiyah:dedekind}.
\end{remark}

If in \eqref{eq:det-transf-iso}, $\spinbasic \otimes F_{w}$ is an odd
spin structure and $z=0$ then we expect that $e^{i\xi(w;
\Lambda)}$ is not well defined because the determinant sections
vanish.  Below, the flat bundle $F_{z}$ in
\eqref{eq:det-transf-iso} is used to determine $e^{i\xi(w;
\Lambda)}$ via a limit process for $\spinbasic \otimes F_{w}$ an odd
spin structure. First we note that
\begin{align}
    \frac{(\detp \dbar)(\spinbasic \otimes F_{\delta})}{N_{\delta}^{2}} &=
    \lim_{z\to 0} \frac{(\det\dbar)( \spinbasic \otimes F_{\delta} \otimes 
    F_{z})}{N_{\delta}^{2}\; \sqrt{\lambda(z)}}\,,
    \nonumber \\
    & =  \frac{1}{\pi} \;
    \frac{(\det\dbar)(\spinbasic)}{ \vartheta(0 ; \Omega)}\,,
    \label{eq:detp-N}
\end{align}
where we used \eqref{eq:detp-section}. Next, 
using the notation $F_{w'} = F_{t} \otimes F_{f^{*}w}$ and $F_{z'} =
F_{f^{*}z}$ we compute
\begin{align*}
    \frac{(\detp \dbar)(\spinbasic' \otimes F_{w'})}{N_{\delta'}^{2}} &=
     \lim_{z'\to 0}
    \frac{(\det\dbar)( \spinbasic' \otimes F_{w'} \otimes
    F_{z'})}{N_{\delta'}^{2}\; \sqrt{\lambda'(z')}}\,, \\
    & =
     \lim_{z'\to 0}
    \frac{e^{i\xi(\delta,\Lambda)}\; (\det \dbar)(\spinbasic \otimes F_{w} 
    \otimes F_{z})}{\pi\; \sum z'_{j} \; \partial 
    \vartheta[\delta'](0;\Omega')/\partial z'_{j}} \,, \\
    & = 
     \lim_{z\to 0}
    \frac{e^{i\xi(\delta,\Lambda)}\; (\det \dbar)(\spinbasic \otimes F_{w} 
    \otimes F_{z})}{\pi\; \varepsilon(\Lambda) e^{-i\pi 
    \phi(\delta,\Lambda)} \det(C\Omega + D)^{1/2} \sum z_{j} \; \partial 
    \vartheta[\delta](0;\Omega)/\partial z_{j}} \,, \\
    & =
    \frac{e^{i\xi(\delta,\Lambda)}}{\varepsilon(\Lambda)
    e^{-i\pi \phi(\delta,\Lambda)} \det(C\Omega + D)^{1/2}}
    \lim_{z\to 0}
    \frac{(\det \dbar)(\spinbasic \otimes F_{w} 
    \otimes F_{z})}{N_{\delta}^{2}\; \sqrt{\lambda(z)}} \,, \\
    &=
    \frac{e^{i\xi(\delta,\Lambda)}}{ \varepsilon(\Lambda)
    e^{-i\pi \phi(\delta,\Lambda)} \det(C\Omega + D)^{1/2}}   
    \frac{(\detp \dbar)(\spinbasic \otimes F_{w})}{N_{\delta}^{2}} \,.
\end{align*}
Hence
\begin{equation}
    \frac{(\detp \dbar)(\spinbasic' \otimes F_{w'})}{N_{\delta'}^{2}} =
    e^{i\xi(\delta,\Lambda)} \varepsilon(\Lambda)^{-1} e^{i\pi \phi(\delta,\Lambda)}
    \det(C\Omega+D)^{-1/2} \; \frac{(\detp \dbar)(\spinbasic \otimes
    F_{w})}{N_{\delta}^{2}}\;.
    \label{eq:detp-transf}
\end{equation}
Combining the above with \eqref{eq:symp-hdelta} gives the interesting
``isometry'' relation that
\begin{equation}
    (\detp \dbar)(\spinbasic' \otimes F_{w'}) 
    \left(\frac{h'_{\delta'}}{N_{\delta'}}\right)^{2} = 
    e^{i\xi(\delta,\Lambda)}\;
    (\detp \dbar)(\spinbasic \otimes F_{w}) 
	\left(\frac{h_{\delta}}{N_{\delta}}\right)^{2}\;.
    \label{eq:detp-invarianc}
\end{equation}

This result can be used to determine $e^{i\xi(\delta,\Lambda)}$ for an
odd semi-characteristic $\delta$:
\begin{equation}
    e^{i\xi(0,\Lambda)} = e^{i\xi(\delta,\Lambda)}\; e^{i\pi 
    \phi(\delta,\Lambda)}\,.
    \label{eq:xi-relations}
\end{equation}
\begin{proof}
Use \eqref{eq:detp-section} and \eqref{eq:det-t-1} to observe
\begin{align}
    \frac{(\detp \dbar)(\spinbasic' \otimes F_{w'})}{N_{\delta'}^{2}} 
    & =  \frac{1}{\pi} \;
    \frac{(\det\dbar)(\spinbasic')}{ \vartheta(0 ; \Omega')}\,,
    \nonumber \\
    &=e^{i\xi(0,\Lambda)}\; \varepsilon(\Lambda)^{-1}
     \det(C\Omega+D)^{-1/2}\; \frac{1}{\pi} \;
    \frac{(\det\dbar)(\spinbasic)}{\vartheta(0;\Omega)} \,,
    \nonumber \\
     &=e^{i\xi(0,\Lambda)}\; \varepsilon(\Lambda)^{-1}
     \det(C\Omega+D)^{-1/2}\; \frac{(\detp \dbar)(\spinbasic \otimes
     F_{w})}{N_{\delta}^{2}}\,.
     \label{eq:detp-transf-2}
\end{align}
Comparing \eqref{eq:detp-transf-2} and \eqref{eq:detp-transf} yields 
the desired result \eqref{eq:xi-relations}.
\end{proof}

\subsection{Computing $\xi(0,\Lambda)$}
\label{sec:xi}

In this discussion we consider the conformal field theory with two
$(1/2,0)$ chiral fermions.  The ``partition function'' $Z$ of this CFT
is the determinant section of the determinant line bundle over spin
moduli space.  The pullback section $\widetilde{Z}$ has the
property that 
\begin{equation}
    \widetilde{Z}(\spin) = \widetilde{Z}(f\cdot \spin)
    \label{eq:Z-per}
\end{equation}
where $\spin \in \Teichspin(\Sigma)$ and $f \in \Diff(\Sigma)$.

We consider the case where we have a diffeomorphism $f$ that
represents\footnote{We have learned from Igusa~\cite{Igusa:book} and
Mumford~\cite{Mumford:1} that using $\Gamma_{1,2}$ simplifies life.}
$\Lambda \in \Gamma_{1,2}$.  This diffeomorphism maps
$\Teichspin(\Sigma,\qbasic)$ to itself because of the definition of
$\Gamma_{1,2}$.  Let $\spinbasic \in \Teichspin(\Sigma,\qbasic)$ then
using \eqref{eq:Z-per} and trivialization \eqref{eq:Teich-triv} we
conclude that
\begin{equation*}
    \det\dbar(f\cdot\spinbasic) \exp\left( -
    \int_{\spinbasic_{0}}^{f\cdot\spinbasic} \widetilde{\mathcal{A}}
    \right) \hat{\sigma}_{0} = \det\dbar(\spinbasic) \exp\left( -
    \int_{\spinbasic_{0}}^{\spinbasic} \widetilde{\mathcal{A}} \right)
    \hat{\sigma}_{0}\,.
\end{equation*}
We obtain the result
\begin{equation}
	\det\dbar(f\cdot\spinbasic) = \det\dbar(\spinbasic) \exp\left(
	\int_{\spinbasic}^{f\cdot\spinbasic} \widetilde{\mathcal{A}}
	\right)
    \label{eq:xi0-12}
\end{equation}
Because of flatness we have  that
\begin{equation*}
    \exp\left( \int_{\spinbasic}^{f\cdot\spinbasic}
    \widetilde{\mathcal{A}} \right) = \exp\left(
    \int_{\spinbasic_{0}}^{f\cdot\spinbasic_{0}} \widetilde{\mathcal{A}}
    \right),
\end{equation*}
\emph{i.e.}, the parallel transport factor is independent of
basepoint.  Any path from $\spinbasic_{0}$ to $f\cdot\spinbasic_{0}$
projects to a loop in $\Modonetwo(\Sigma)$, where $\Modonetwo(\Sigma)
= \Teich(\Sigma)/\Diff_{1,2}(\Sigma)$ and $\Diff_{1,2}(\Sigma)$ are
the diffeomorphisms representing $\Gamma_{1,2}$.  We see that the
exponential above is just the holonomy with respect to the flat
connection $\mathcal{A}$ pulled back to the pull back line bundle over
$\Modonetwo(\Sigma) \supset \Mod(\Sigma)$.  We know that the flat
hermitian connection gives a homomorphism $\hol:
\pi_{1}(\Modonetwo(\Sigma))\to S^{1}$ and by
relating \eqref{eq:xi0-12} with \eqref{eq:det-transf-iso} we obtain
\begin{thm}\label{thm:hol-12}
    If $\Lambda \in \Gamma_{1,2}$ then
	\begin{equation}
	    e^{i\xi(0,\Lambda)} = \hol(\Lambda)^{-1}\,.
	    \label{eq:xi0-12-final}
	\end{equation}
    The holonomy is independent of the choice of representative $f\in 
    \Diff(\Sigma)$ chosen for $\Lambda$.
\end{thm}
\begin{remark}\label{rem:torus}
    Because $S^{1}$ is abelian, the set of maps of 
    $\pi_{1}(\Modonetwo(\Sigma) \to S^{1}$ is the same as the set of 
    maps $H_{1}(\Modonetwo(\Sigma),\bbZ) \to S^{1}$, \emph{i.e.}, the 
    character group of $H_{1}(\Modonetwo(\Sigma),\bbZ)$ which equals 
    $H^{1}(\Modonetwo(\Sigma),\bbR)/H^{1}(\Modonetwo(\Sigma),\bbZ)$, 
    a torus $T$.    
    Therefore our map $\hol$ is a homomorphism from $\Gamma_{1,2}$ to
    the torus $T$.  Our theorem states that $e^{i\xi(0,\Lambda)} =
    \hol(\Lambda)^{-1} \in T$.  We will denote $e^{i\xi(0,\Lambda)}$
    by $t_{\Lambda}\in T$.  We can also interpret
    $e^{in\xi(0,\Lambda)}$ by defining $\hol^{n}: \Gamma_{1,2} \to T$
    by $\hol^{n}(\Lambda) = \left( \hol(\Lambda)\right)^{n}$ so that
    $t_{\Lambda}^{n} = e^{in\xi(0,\Lambda)}= \left(
    \hol^{n}(\Lambda)\right)^{-1} \in T$.
\end{remark}

\begin{remark}
    The original proposal relating global anomalies to holonomy is 
    due to Witten \cite{Witten:1985xe}.
\end{remark}

\subsection{Main Results}
\label{sec:results}

As we noted in the introduction, the purpose of this paper is to 
extend the results of BEG, where we had a theory over Teichmuller 
space, to a theory over moduli space. 

We will denote the dependence of $\Zsc$ on a string manifold $M$ and 
a spin structure $\spin$ by $\Zsc(M,\spin)$.

\begin{remark}
    $\Zsc(M,\spin)$ is a function on $\Teichspin(\Sigma)$ and the map 
    $M \mapsto \Zsc(M,\spin)$ is a homomorphism from string cobordism 
    $\MString^{*}$ to functions on $\Teichspin(\Sigma)$, \emph{i.e.}, 
    a genus.
\end{remark}

It turns out it is useful to define a function $\Phi$ by
\begin{equation}
    \Zsc(M,\spin) =\left(
    \frac{(\det\dbar)(\spinbasic)}{\vartheta(0 ;\Omega)} \right)^{2n} 
    \Phi(M,\spin)\,.
    \label{eq:Zsc-2}
\end{equation}
The map $M \mapsto \Phi(M,\spin)$ is also a genus from $\MString^{*}$ 
to $\Teichspin(\Sigma)$.

Formula \eqref{eq:def-Zsc-1} for $\Zsc$ contained $\detp \partial$ and
not $\detp\dbar$, therefore we have to take the complex conjugate.
Using the  notation of \eqref{eq:def-Zsc-1} we find the
transformation law
\begin{equation}
    \left(\frac{\detp \partial_{\delta'}}{\Nb_{\delta'}^{2}}
    \right)_{\bm{g}'} = \varepsilon(\Lambda)
    e^{-i\xi(0,\Lambda)}\;
    \det(C\bar{\Omega} +D)^{-1/2} \; \left(\frac{\detp
    \partial_{\delta}}{\Nb_{\delta}^{2}} \right)_{\bm{g}}.
    \label{eq:detp-transf-1}
\end{equation}

A drawback of Theorem~\ref{thm:transf-integral} is that the integral
is a combination of objects that vary holomorphically and objects that
vary anti-holomorphically with respect to Teichmuller space.  For
example $\Omega$ varies holomorphically while $\hb^{2}_{\delta}$
varies anti-holomorphically.  Next we try to group terms together in
such a way as to try to be as holomorphic as possible.  In particular
we remind the reader that our path integral viewpoint acts as a guide
to our final result.  We use the Belavin-Knizhnik theorem to formally
factorize the determinants and also we invoke \eqref{eq:detp-section}.
With this in mind we make two substitutions in \eqref{eq:def-Zsc-1}:
\begin{align}
    \left( \frac{\vol \Sigma \det H}{\detz} \right)^{1/2} & \longrightarrow 
    \frac{1}{2}\; \frac{(\det\dbar)(\spinbasic)}{\vartheta(0 ;\Omega)} \;
    \frac{(\det\partial)(\spinbasic)}{\overline{\vartheta(0 ;\Omega)}}\;,
    \label{eq:rep-1} \\
    \frac{\detp \partial_{\delta}}{\bar{N}_{\delta}^{2}} & \longrightarrow
    \frac{1}{\pi}\; \frac{(\det
    \partial)(\spinbasic)}{\overline{\vartheta(0;\Omega)}} \,.    
    \label{eq:rep-2}
\end{align}
The expression for $\Phi$ now becomes
\begin{equation}
    \Phi(M,\spin) =  \frac{1}{(4\pi)^{n}}\; \frac{1}{(\det H)^{n}}\; \left(
    \frac{(\det\partial)(\spinbasic)}{\overline{\vartheta(0 ;\Omega)}}
    \right)^{3n} \int_{M} \prod_{r=1}^{n}
    \frac{\zanalog\left(x_{r}\; z\left( \hb_{\delta}^{2} \right);
    \Omega \right)}{\vartheta[\kappa]\left(x_{r}\; z\left(
    \hb_{\delta}^{2} \right); \Omega \right)}\;.
    \label{eq:def-Phi-1}
\end{equation}
Using \eqref{eq:symp-hodge}, \eqref{eq:det-t-1} and 
\eqref{eq:transf-theta} we obtain

\begin{thm}
    \label{thm:transf-Phi}
    Let $M$ be a string manifold  with
    $\dim M=2n$; if $\Lambda \in \Sp(2g,\bbZ)$ is represented by $f
    \in \Diff(\Sigma)$, then under the action of $f$ we have
\begin{subequations}\label{eq:Phi-transf}
    \begin{align}
    \Phi' &= \varepsilon(\Lambda)^{2n}\, e^{-3ni
    \xi(0,\Lambda)}\; e^{i n \pi \phi(\delta,\Lambda)}\;
    \det(C\Omega+D)^{n}\; \Phi\,,
    \label{eq:Phi-transf-a} \\
    &= \varepsilon(\Lambda)^{2n}\, e^{-3ni
    \xi(\delta,\Lambda)}\; e^{-2 i n \pi \phi(\delta,\Lambda)}\;
    \det(C\Omega+D)^{n}\; \Phi\,.
    \label{eq:Phi-transf-b}
\end{align}
\end{subequations}
\end{thm}

\begin{cor}
    \label{thm:transf-Zsc}
    Let $M$ be a string manifold  with
    $\dim M=2n$; if $\Lambda \in \Sp(2g,\bbZ)$ is represented by $f
    \in \Diff(\Sigma)$, then under the action of $f$ we have
\begin{subequations}\label{eq:Zsc-transf}
    \begin{align}
    \Zsc' &= e^{-in\xi(\delta,\Lambda)}\; \Zsc\,,
    \label{eq:Zsc-transf-a} \\
    &= e^{-in \xi(0,\Lambda)}\; e^{ i n \pi \phi(\delta,\Lambda)}\;
    \Zsc\,.
    \label{eq:Zsc-transf-b}
\end{align}
\end{subequations}
\end{cor}

If $\Lambda \in \Gamma_{1,2}$ then for a semi-characteristic
$\left[\begin{smallmatrix} a \\ b \end{smallmatrix}\right]$ we have
that $\phi(a,b,\Lambda) = \frac{1}{2} \bmod \bbZ$.  If $\dim M = 2n =
4k$ then $e^{-2 i n \pi \phi(\delta,\Lambda)} =1$, and the spin
structure $\spinbasic$ maps onto the preferred spin structure $\spinbasic'$.
According to Igusa~\cite[p.~181]{Igusa:book}, $\varepsilon(\Lambda)^{2}$ is a
character of $\Gamma_{1,2}$.  Note that $\varepsilon(\Lambda)^{2n} =
\varepsilon(\Lambda)^{4k} = \pm 1$.

\begin{thm}
    \label{thm:transf-Phi-G12}
    Let $M$ be a string manifold  with
    $\dim M=2n=4k$; if $\Lambda \in \Gamma_{1,2} \subset \Sp(2g,\bbZ)$
    is represented by $f \in \Diff(\Sigma)$, then under the action of
    $f$ we have
    \begin{subequations}\label{eq:Phi-transf-G12}
	\begin{align}
	\Phi' &= \varepsilon(\Lambda)^{2n}\, e^{-3ni
	\xi(0,\Lambda)}\; e^{i n \pi \phi(\delta,\Lambda)}\;
	\det(C\Omega+D)^{n}\; \Phi\,,
	\label{eq:Phi-transf-a-G12} \\
	&= \varepsilon(\Lambda)^{2n}\, e^{-3ni
	\xi(\delta,\Lambda)}\;
	\det(C\Omega+D)^{n}\; \Phi\,.
	\label{eq:Phi-transf-b-G12}
    \end{align}
    \end{subequations}
    Note that $\varepsilon(\Lambda)^{2n} = \pm 1$ and $e^{i n \pi
    \phi(\delta,\Lambda)} = \pm 1$.
\end{thm}

If $\dim M = 2n = 8l$ then $e^{in\pi \phi(\delta,\Lambda)}=1$,
$e^{in\pi \xi(\delta,\Lambda)} = e^{in\pi \xi(0,\Lambda)}$, and 
$\varepsilon(\Lambda)^{2n}=1$; and we find
\begin{thm}
    \label{thm:transf-Phi-G12-8}
    Let $M$ be a string manifold  with
    $\dim M=2n=8l$; if $\Lambda \in \Gamma_{1,2} \subset \Sp(2g,\bbZ)$
    is represented by $f \in \Diff(\Sigma)$, then under the action of
    $f$ we have
    \begin{align}
	\Phi' &= e^{-3n i \xi(0,\Lambda)}\; \det(C\Omega+D)^{n}\;
	\Phi\,,
	\label{eq:Phi-transf-b-G12-8} \\
	\Zsc' & = e^{-in\xi(0,\Lambda)}\; \Zsc \,.
	\label{eq:Zsc-transf-b-G12-8}
    \end{align}
    The transformation laws above do not depend on the choice of spin 
    structure $\delta$.
\end{thm}

\begin{remark}
    We interpret \eqref{eq:Zsc-transf-b-G12-8} as saying that $\Zsc$ 
    is a modular function with multiplier $t_{\Lambda}^{-n}$; see 
    Remark~\ref{rem:torus} for notation.
\end{remark}
\begin{remark}
    We interpret \eqref{eq:Phi-transf-b-G12-8} as saying that $\Phi$ 
    is a modular form with weight $n$ and with multiplier 
    $t_{\Lambda}^{-3n}$.
\end{remark}

\section*{Acknowledgments}
\label{sec:acknowledgments}%
\addcontentsline{toc}{section}{Acknowledgments}%

We wish to thank M.J. Hopkins for introducing us to quarfs and their
relation to spin structures, for explaining the various cobordism
theories, and for encouragement in completing this project.  The work
of OA was supported in part by the National Science Foundation under
grants PHY-0554821 and PHY-0854366.  The work of IMS was supported by
two DARPA grants through the Air Force Office of Scientific Research
(AFOSR): grant numbers FA9550-07-1-0555 and HR0011-10-1-0054.

\appendix

\section{Isomorphism between $\Lhz$ and $\Lho$}
\label{sec:conventions}

On $\Sigma$ we have local complex coordinates $z = x+ iy$.  In
isothermal coordinates the metric is $ds^{2} = g_{z\zbar} (dz \otimes
d\zbar + d\zbar \otimes dz)$.  The volume element is $\sqrt{\det g}\;
dx \wedge dy$.  We note that $d\zbar \wedge dz = 2i\, dx \wedge dy$.
If $\omega_{i}$ are the standard abelian differentials then
$\int_{\Sigma} \bar{\omega}_{i} \wedge \omega_{j} = (\Omega -
\bar{\Omega})_{ij}$.

If $K$ is the canonical bundle of $\Sigma$ and if $\xi (dz)^{n}$ and 
$\eta (dz)^{n}$ are section of $K^{n}$ then the inner product is 
defined by
\begin{equation}
    \langle \xi (dz)^{n} , \eta (dz)^{n} \rangle = \int_{\Sigma}
    \left( g^{\zbar z}\right)^{n} \; \bar{\xi} \eta \; \sqrt{\det g}\; dx
    \wedge dy\,.
    \label{eq:def-ip-Kn}
\end{equation}

The explicit conjugate linear isomorphism between $\Lho$ and $\Lhz$ is
constructed as follows.  Let $h$ be a section of $\Lhz$ then the
corresponding section $\psi$ of $\Lho$ is given by
\begin{equation}
    \psi =   (g_{z\zbar}\; dz \otimes d\zbar)^{1/2} \; \bar{h}\,.
    \label{eq:def-iso-conj}
\end{equation}

\section{Identifying a Holomorphic Cross Section of $\cL$ with a $\vartheta$ Function}
\label{sec:app-kappa}

Our purpose is to identify a holomorphic cross section (unique up to
scale) of a determinant line bundle with a theta function.  We also
set the notation for the algebraic geometry we need and provide an
introduction for non-experts.

In this appendix we make the identification using facts about Riemann 
surfaces well known to algebraic geometers.

Let $J_r(\Sigma)$ denote the set of holomorphic line bundles over 
$\Sigma$ with $\cclass_1$ equal to $r$. Of particular interest to us 
is $J_0(\Sigma)$ (our $J(\Sigma)$), the set of flat line bundles, 
which we can identify with $\widetilde{\pi_1(\Sigma)} = \{ 
\chi:\pi_1(\Sigma) \to S^1\} \simeq \widetilde{H_1(\Sigma,\bbZ)}
\simeq H^1(\Sigma,\bbR)/H^1(\Sigma,\bbZ)$. The last isomorphism can 
be described in terms of (real) closed $1$-forms $\omega$: let 
$\chi_\omega(\gamma) = \exp\left(2\pi i \int_\gamma \omega\right)$ 
for $\gamma$ a closed path starting at $P_0\in\Sigma$. Clearly 
$\chi_\omega$ is a homomorphism $\pi_1(\Sigma)\to S^1$ and depends 
only on the cohomology class of $\omega$. It is easy to see that 
$\omega \mapsto \chi_\omega$ induces an isomorphism 
$H^1(\Sigma,\bbR)/H^1(\Sigma,\bbZ) \to \widetilde{\pi_1(\Sigma)}$. 
Here we take $\pi_1(\Sigma)$ as the closed (piecewise smooth) paths 
starting at $P_0$ with equivalence relationship given by homotopy.

When the real surface $\Sigma$ has a complex structure, then
$H^1(\Sigma,\bbC) = H^1(\Sigma,\bbR)\otimes\bbC \simeq \Hoz \oplus
\Hzo$.  Taking the real part induces an isomorphism of $\Hzo$ (or
$\Hoz$) with $H^1(\Sigma,\bbR)$.  Let $\real: \Hzo \to
H^1(\Sigma,\bbR)$ be this isomorphism, and let $\widetilde{\Hzo} =
(\real)^{-1} H^1(\Sigma,\bbZ)$ so that
$H^1(\Sigma,\bbR)/H^1(\Sigma,\bbZ) \simeq \Hzo/\widetilde{\Hzo}$. 
Since $\Hzo$ is a complex vector space, $\Hzo/\widetilde{\Hzo}$ is a
complex torus $J(\Sigma)$, the jacobian of $\Sigma$.  Our chain of
arguments demonstrates that the jacobian $J_0(\Sigma)$ is isomorphic
to $\Honehat$, the character group of $H_{1}(\Sigma,\bbZ)$. 
Specifically, if $\mu \in \Hzo$, let
\begin{equation}
	\chi_\mu(\gamma) = e^{2\pi i \int_\gamma(\mu 
	+\bar{\mu})/2}\;,
	\label{chi-mu}
\end{equation}
with $\gamma$ a loop with basepoint $P_0$, then $\mu\to \chi_\mu$
induces the isomorphism of $\Hzo/L_\Omega = J(\Sigma)$ with
$\Honehat$.  $L_{\Omega} =\widetilde{\Hzo}$ will be described
differently in the paragraph below.  The covering space of
$J(\Sigma)=J_0(\Sigma)$ is $\Hzo$.

We chose a standard basis $(\bar{\omega}_{1},\ldots,\bar{\omega}_{g})$
of $\Hzo$ obtained from a choice of symplectic basis
$(\ba_{j},\bb_{k})$ of $H_{1}(\Sigma,\bbR)$ so that
$\int_{\ba_{i}}\omega_{j}=\delta_{ij}$.  We remind the reader that the
Riemann period matrix $\Omega_{ij}= \int_{\bb_{i}}\omega_{j}$ with
imaginary part $(\Omega_{ij}-\bar{\Omega}_{ij})/2i$ that is positive
definite and in fact equal to $\langle \omega_{i},\omega_{j}\rangle$.
Then $\omega_i = \sum\alpha_i + \Omega_{ij} \beta_j$ where $\alpha$
and $\beta$ are the harmonic representatives dual to the $\ba$ and
$\bb$ cycles.  We write a point in $\Hzo$ as $(u_j + i
v_j)\bar{\omega}_j$ where $u$ and $v$ are real.  One can easily verify
that $\widetilde{\Hzo}$ is represented by $2\sum(m + \Omega
n)_j(\Omega-\bar{\Omega})_{jk}^{-1}\bar{\omega}_k$ where $m \in\bbZ^g$
and $n\in\bbZ^g${}\footnote{The standard basis $\{\omega_{j}\}$ of
$\Hoz$ identifies $\Hoz$ with $\bbC^{g}$. Algebraic geometers
\cite[p.~143]{Mumford:1} identify $H_{1}(\Sigma,\bbC)$ with $\bbC^{g}$
using the Abel map.  As a consequence the lattice $L_{\Omega} \subset
\bbC^{g}$ is the dual torus to the algebraic geometers' Jacobian
torus.}.

Multiplication $m_L$ by a line bundle $L \in J_r(\Sigma)$ gives an 
isomorphism $m_L : J_0(\Sigma) \to J_r(\Sigma)$. In particular a spin 
structure $\sqrt{K}$, where $K$ is the canonical bundle of $\Sigma$, 
gives $m_{\sqrt{K}}: J_0(\Sigma) \to J_{g-1}(\Sigma)$ with $g$ the 
genus of $\Sigma$. Similarly for $P_0\in\Sigma$ let $\Lpz$ be the 
line bundle with divisor $P_0$ so that $\Lpz \in J_1(\Sigma)$. Then 
$m_{\Lpz^r}: J_0(\Sigma) \to J_r(\Sigma)$ is an isomorphism. The 
complex structure on $J_r$ is chosen such that $m_L$ is holomorphic.

One can construct a Poincar\'{e} line bundle $Q_r$ over 
$J_r(\Sigma)\times\Sigma$ whose restriction to each fiber $\{L\}\times 
\Sigma$ is the line $L\in J_r(\Sigma)$.  The holomorphic line bundle 
$Q_r$ over $J_r(\Sigma)\times\Sigma$ is determined only up to a line 
bundle on $J_r(\Sigma)$ pulled up to $J_r(\Sigma)\times\Sigma$. A 
choice of point $P_0\in\Sigma$ determines $Q_r$ by stipulating that 
$Q_r|_{J_r(\Sigma)\times\{P_0\}} \simeq 1$ on $J_r(\Sigma)$.
In \BEG{Appendix~B} we construct such a $Q_0$ explicitly. 
We can use $(m_{L_{P_{0}}^{g-1}})^{*} Q_{g-1}$ instead.

Let $\dbar\otimes I_{Q_r}$ be the family of $\dbar$ operators 
parametrized by $Q_r$.  Suppose $\cM$ is a holomorphic line bundle on 
$J_r(\Sigma)$ which pulled up to $J_r(\Sigma)\times\Sigma$ is $\cMt$.  
Suppose we have modified our choice of Poincar\'{e} line bundle $Q_r$ 
by $Q_r\otimes\cMt$.  One can show the determinant line bundle of the 
family $\dbar\otimes I_{Q_r\otimes\cMt}$, $\DET(\dbar\otimes 
I_{Q_r\otimes\cMt})$ is isomorphic to $\DET(\dbar\otimes I_{Q_r})
\otimes \cM^{r+1-g}$.  In particular, when $r=g-1$, $\DET(\dbar\otimes 
I_{Q_{g-1}})$ is independent of choice of $Q_{g-1}$.

The choice of $r=g-1$ is special because the index of the operator
$\dbar\otimes I_L$, $L\in Q_{g-1}$, is zero.  Generically the operator
$\dbar\otimes I_L$ is invertible.  Let $\cV = \{ L\in J_{g-1}(\Sigma)
\;|\; \dbar\otimes I_L \quad\mbox{is not invertible}\}$.  $\cV$ is a
variety in $J_{g-1}(\Sigma)$, in fact the divisor of the line bundle
$\DET(\dbar\otimes I_{Q_{g-1}})$.  Of course, $\cV$ is also $\{ L\in
J_{g-1}(\Sigma) \;|\; L \quad\mbox{has a nonzero holomorphic
section}\}$. Another description of $\dbar \otimes I_{Q_{g-1}}$ is 
obtained by choosing a spin structure, a $\sqrt{K}$, which we denote 
by $\Lhz$. $m_{\sqrt{K}}$ maps $\Jacz$ to $J_{g-1}(\Sigma)$ and the 
family becomes a family over $\Jacz$, namely $\dbar: \Lhz \otimes 
F_{\chi} \to \Lho \otimes F_{\chi}$ with $F_{\chi} \in \Jacz$.

We use $m_{L^{g-1}_{P_0}}$ to compare $\DET(\dbar\otimes I_{Q_{g-1}})$ 
with $\DET(\dbar\otimes I_{Q_0})$, the latter our line bundle $\cL$ 
over $J_0(\Sigma)$. The Grothendieck-Riemann-Roch theorem implies that
$$
	\left(m_{L^{g-1}_{P_0}}\right)^*\left(\DET(\dbar \otimes 
	I_{Q_{g-1}})\right)
$$
is isomorphic to $\cL$.  Although it is well known that
$H^0(\cL,\bbC)$ has complex dimension one, i.e., the holomorphic
sections of $\cL$ form a one dimensional subspace \BEG{Appendix~C}.
We now want to identify a properly normalized holomorphic section of
$\cL$ with a $\vartheta$-function.

First identify $\widetilde{J_0(\Sigma)}$, the universal cover of
$J_0(\Sigma)$, with $\bbC^g$ using the basis $\{\bar{\omega}_{j}\}$,
where we have defined the Riemann theta function $\vartheta$ and its
divisor.  Let $\cL_\vartheta$ be the holomorphic line bundle over
$J_0(\Sigma)$ whose divisor pulls up to the divisor of $\vartheta$.
We learn from Riemann surface theory that there exists a spin
structure $\sqrt{K}\in J_{g-1}(\Sigma)$ such that $\cL_\vartheta =
m^*_{\sqrt{K}} \DET(\dbar\otimes I_{Q_{g-1}})$.  Put another way,
$m_{\sqrt{K}}(\mbox{divisor of $\cL_\vartheta$}) = \cV$.  The spin
structure is the one determined by the choice of symplectic basis of
cycles in $H_{1}(\Sigma,\bbR)$. In fact, the spin structure 
$\sqrt{K}$ has quarf the basic quarf $\qbasic$ described in 
Section~\ref{sec:quarf-spin}. See also \cite[pp.~162]{Mumford:1}.

Putting these two facts together gives $\cL \simeq
\left(m^{-1}_{\sqrt{K}} m_{L^{g-1}_{P_0}}\right)^* \cL_\vartheta
\simeq \left( m_{K^{-1/2} L^{g-1}_{P_0}}\right) \cL_\vartheta$.  The
flat line bundle $K^{-1/2} L^{g-1}_{P_0}$ lies in $J_0(\Sigma)$ and is
in fact $\kappa$ where $\kappa$ is the Riemann constant, see
\cite[p.~166]{Mumford:1} or \cite[p.  338]{Griffiths:Harris}.  Hence
$\cL$ is the translate of $\cL_\vartheta$ by $\kappa$.  As a result,
our $\vartheta\left[\begin{smallmatrix} a \\ b
\end{smallmatrix}\right]$ is the translate of $\vartheta$ by
$\kappa$; the characteristic $\left[\begin{smallmatrix} a \\ b
\end{smallmatrix}\right]$ equals $\kappa$.

\begin{lem}
    \label{thm:theta-kap-0}
    $\vartheta[\kappa](0)=0$
\end{lem}
\begin{proof}
    This is a consequence of Riemann's Theorem \cite[Corollary 3.6 on
    p.~160]{Mumford:1} or \cite[p.~338]{Griffiths:Harris}.  In
    Mumford's conventions\footnote{Mumford's vector of Riemann
    constants is the negative of Griffiths and Harris' vector of
    Riemann constants.} we have that $\Theta = W_{g-1} - \kappa$.  The
    divisor $W_{g-1}$ is the image of $\Sym^{g-1}(\Sigma)$ under the
    Abel map, \emph{i.e.}, the image in $\Jacz$ of line bundles $L$ of
    Chern class $g-1$ \emph{that have a holomorphic
    section}\footnote{The image under the map $L \to L \otimes
    (L_{P_{0}}^{g-1})^{-1}$}.  Note that the origin $O \in \Jacz$ is
    in $W_{g-1}$ (it is the image of $(P_{0},\dotsc,P_{0})$).
    Riemann's Theorem then tells us that $-\kappa \in \Theta$ but
    $\Theta$ is symmetric therefore $\kappa \in \Theta$.  We have that
    $\vartheta[\pm\kappa](0) \propto \vartheta(\pm\kappa) = 
    \vartheta(\kappa)=0$.
\end{proof}
  
Note that $\vartheta[\kappa](0)=0$ if $\kappa$ is an odd
characteristic.  Mumford \cite[p.~3.82]{Mumford:2}, and Farkas and
Kra~\cite{Farkas:Kra} discuss that in hyperelliptic surfaces, $\kappa$
will be even or odd depending on the genus.  If $\kappa$ is even then
$\vartheta[\kappa](\cdot)$ has a double zero at the origin and this
has implications on the number of zero modes of the Dirac operator.
All genus $2$ surfaces are hyperelliptic and $\kappa$ is an  odd
characteristic.

\subsection{Definition of $\mathfrak{z}_{\kappa}$}
\label{sec:def-z}

The technicalities in constructing the determinant line
bundle are related to the jump in the dimensionality of $\ker D$ at
the origin $O \in\Jacz$: 
\begin{thm}[\BEG{Theorem~B.1}]\label{thm:BEG-B-1}
    Let $\chi \in J(\Sigma)$ be a character and let $\Azo$ be the 
    associated flat connection.  If $\chi = 1$ then $\dim \ker D_0 = 
    g$.  If $\chi \neq 1$ then $\dim \ker D_A = g-1$.
\end{thm}
We can exploit this theorem to give a global construction of the
determinant line bundle over $\Jacz$.  Note that $\Index D_{A} =g-1$
thus: (1) if $\chi \neq 1$ then $D_{A}$ is surjective, (2) if $\chi=1$
then the image of $D_{0}$ is $\onep$.  This motivates us to define a
modified family of operators (parametrized by flat connections) by
$\Dtil_{A}: \Loz \to \onep$ where we have ``killed the constants'' in
the range.  Hence, these operators have numerical index $g$ and
$\dim\coker\Dtil_{A}=0$.  Consequently $\ker \Dtil_{A} = g$ and the
kernel of $\Dtil$ is an honest holomorphic vector bundle over $\Jacz$.
The top wedge power of this vector bundle is the determinant line
bundle $\DET \Dtil \to \Jacz$ and it is isomorphic to $\DET D \to
\Jacz$. If we fix the metric then this construction is holomorphic 
because the $\onep$ does not vary as we move over $\Jacz$.

Our formula for $\Zsc$ requires a specific trivialization of the
determinant line bundle $\DET D \to \Jacz$.  Here we briefly review
the detailed construction given in \BEG{Section~7}.

In a small neighborhood $U$ of the origin we introduce the operator
$\cDt: \Loz \to \Lzz$ defined\footnote{Note that the image of $\cDt$ 
is contained in $\onep$ by construction.} by
\begin{equation}
    \cDt\phi = D\phi -i\int_{\Sigma} \Azo \wedge \phi
    \label{eq:def-cDt}
\end{equation}
that has the property that $\dim\ker \cDt =g$ on $U$ and we construct
$\DET \cDt \to U$.  Note that both $\cDt$ and $D$ are operators with
index equal to $g-1$.  Let $V = \Jacz - \{O\}$ then $\dim \ker D =
g-1$ on $V$, and we construct $\DET D \to V$.  On $U \cap V$ we have
that $\ker D \subset \ker \cDt$ and an exact sequence of vector
bundles
\begin{equation*}
    0 \to \ker D \to \ker \cDt \to \mathcal{K} \to 0\,.
\end{equation*}
The holomorphic line bundle $\mathcal{K}$ is used to patch $\DET D$
with $\DET \cDt$ on $U \cap V$, and gives a construction of the
holomorphic determinant line bundle $\DET D \to \Jacz$.

In \BEG{Sections~7~\&~8} we discussed a convenient trivialization of $\DET
D$ over $U$ by trivializing $\DET \cDt$ over $U$.  Note that $\ker \cDt
\rvert_{O} = \Hoz$ and as a consequence of the Hodge theorem we can
choose a basis $(\varpi_{1},\dotsc, \varpi_{g})$ for $\ker\cDt$ on $U$
such that $\varpi_{j} = \omega_{j} + \Upsilon_{j}$ where 
$\Upsilon_{j} \in \Hoz^{\perp} \subset \Loz$. In fact \BEG{eq.~(8.1)} 
shows that if we subject a standard basis $(\omega_{1},\dots,\omega_{g})$ of 
$\Hoz$ to a symplectic change of basis, see eq.~\eqref{eq:symp-omega}, 
then $(\varpi_{1},\dotsc, \varpi_{g})$ transforms the same way.

We know \BEG{Appendix~C} that $\DET D \to \Jacz$ has a
unique holomorphic section $\theta_{\kappa}$ up to scale.
Lemma~\ref{thm:theta-kap-0} tells us that this section vanishes at the
origin.  In the $(\varpi_{1},\dotsc, \varpi_{g})$ trivialization we
write this section as
\begin{equation}
    \theta_{\kappa} = \frac{\zanalog}{\varpi_{1}\wedge \dotsb \wedge 
    \varpi_{g}}
    \label{eq:def-z}
\end{equation}
as explained in \BEG{eq.~(8.10)}.  In the trivialization given by
lifting the line bundle to $\Hzo$, the universal cover of $\Jacz$, the
section is given by the theta functions $\vartheta[\kappa](\cdot)$.
Therefore the ratio $\zanalog(z)/\theta[\kappa](z)$ is just the
holomorphic transition function that takes you from one trivialization
of the line bundle to the other.  In \BEG{eq.~(9.5)} we give a formula
for this transition function:
\begin{equation}
    e^{-\pi i \sum z (\Omega - \bar{\Omega})^{-1} z}\;
    \frac{\zanalog(z)}{\vartheta[\kappa](z)} = (\varpi_{1}\wedge\dotsc
    \wedge \varpi_{g}) \otimes e^{-\int_{0}^{A}(\nutil-\rho)}\;
    (\omega_{1}\wedge \dotsb \wedge \omega_{g})^{-1}\,.
    \label{eq:z-theta}
\end{equation}
Note that the right hand side is holomorphic and non-vanishing as
required for a transition function.  Therefore the ratio on the left
hand side is holomorphic in $U$.  
\begin{remark}\label{rem:z-tranf}
An important consequence of the equation above is that the right hand
side of the ratio is invariant under the symplectic action.
\end{remark}

\section{Genus 1 Modular Transformation Examples}
\label{sec:genus-1}

We derive the string genus using \eqref{eq:def-Zsc-1} and specialize
to genus $1$ using the discussion in \BEG{Remark 3}.  The presentation
here is more complete than \BEG{Remark 3} because we are checking our
modularity transformation results for general genus by applying them
to genus one.  There is a single odd spin structure, the theta
function is the odd theta function, and consequently
\begin{equation*}
    \lim_{z\to 0} \frac{\vartheta(z;\tau)}{z} = \vartheta'(0;\tau) = 
    -2\pi \eta(\tau)^{3}\,.
\end{equation*}
We also have $\omega= dz$, $\mathfrak{z} = - \vartheta'(0;\tau) z$,
and we omit the $\delta$ label on the antiholomorphic spinor since
there is only a single odd spin structure $\hb^{2}=
\overline{\vartheta'(0;\tau)}\; d\zbar$.  The modular transformation
$\tau' = -1/\tau$ gives $z'=z/\tau$.  Some examples:
\begin{equation}
    (\hb')^{2} = \overline{\vartheta'(0;\tau')} d\zbar' \propto
    \overline{\tau^{3/2}\vartheta'(0;\tau)}\; d\zbar/\bar{\tau}
    = \bar{\tau}^{1/2}\;\overline{\vartheta'(0;\tau)}\; d\zbar= 
    \bar{\tau}^{1/2}\; \hb^{2}\,,
    \label{eq:h2-g-1}
\end{equation}
in agreement with the complex conjugate of \eqref{eq:symp-hdelta}.
Next we observe that
\begin{equation*}
    z'(\hb') \tau = \left(\tau'_{2} \;\overline{\vartheta'(0;\tau')}
    \right) \tau \propto \tau (\tau_{2}/\tau\bar{\tau})\;
    \overline{\tau^{3/2}\vartheta(0;\tau)} = \bar{\tau}^{1/2} \left(
    \tau_{2}\; \overline{\vartheta(0;\tau)} \right) = \bar{\tau}^{1/2}
    z(\hb)
\end{equation*}
that concurs with \eqref{eq:symp-hp-h}. The main object of interest is
\begin{equation}
    \int_{M} \prod_{r=1}^{n} \frac{\mathfrak{z} \left(x_{r} z(\hb)
    \right)}{\vartheta\left(x_{r} z(\hb); \tau \right)} \propto
    \int_{M} \prod_{r=1}^{n}
    \frac{\vartheta'(0;\tau)}{\vartheta\left(x_{r} z(\hb); \tau
    \right) / \left(x_{r} z(\hb)\right)}.
    \label{eq:last-nice}
\end{equation}
The last expression is reminiscent of the integral in 
\eqref{eq:def-Zsc-1} in the sense that as $x\to 0$  the 
denominator approaches the numerator. The modular 
transformation properties are:
\begin{align*}
    \int_{M} \prod_{r=1}^{n}
    \frac{\vartheta'(0;\tau')}{\vartheta\left(x_{r} z'(\hb'); \tau'
    \right) / \left(x_{r} z'(\hb')\right)} & \propto \int_{M}
    \prod_{r=1}^{n} \frac{\tau^{3/2}\; \vartheta'(0;\tau)}{\left[
    \tau^{1/2} \vartheta\left(x_{r} \bar{\tau}^{1/2}\, z(\hb); \tau
    \right) \right] / \left(x_{r} \bar{\tau}^{1/2}\, z(\hb)/ \tau
    \right)}\\
     & \propto \int_{M} \prod_{r=1}^{n} \frac{\vartheta'(0;\tau)}{
     \vartheta\left(x_{r} \bar{\tau}^{1/2}\, z(\hb); \tau \right) /
     \left(x_{r} \bar{\tau}^{1/2}\, z(\hb) \right)} \\
     & \propto \bar{\tau}^{n/2} \int_{M} \prod_{r=1}^{n}
     \frac{\vartheta'(0;\tau)}{ \vartheta\left(x_{r} z(\hb); \tau
     \right) / \left(x_{r} z(\hb) \right)}
\end{align*}
in agreement with \eqref{eq:transf-theta}. The other 
terms are
\begin{align}
    \frac{\vol \Sigma \det H}{\det_{\onep} \Delta_{0}} &\propto
    \frac{\tau_{2}^{2}}{\det_{\onep} \Delta_{0}} \propto \left[
    \eta(\tau) \overline{\eta(\tau)}\right]^{-2} 
    \label{eq:det0}\\
    \frac{\detp \partial_{1/2}}{N^{2}} &\propto \frac{\tau_{2}
    \overline{\eta(\tau)}^{2}}{\tau_{2} \overline{\vartheta'(0;
    \tau)}} 
    \propto \frac{1}{\overline{\eta(\tau)}}\,.
    \label{eq:det12}
\end{align}
Abusing notation,  $\det H' \propto \tau_{2}' \propto (\det 
H)/(\tau \bar{\tau})$, 
\begin{equation*}
    \left(\frac{\vol \Sigma \det H}{\det_{\onep} \Delta_{0}}
    \right)_{\tau'} = (\tau \bar{\tau})^{-1} \left(\frac{\vol \Sigma
    \det H}{\det_{\onep} \Delta_{0}} \right)_{\tau}\,,
\end{equation*}
and
\begin{equation}
    \left(\frac{\vol \Sigma}{\det_{\onep} \Delta_{0}} \right)_{\tau'}
    = \left(\frac{\vol \Sigma}{\det_{\onep} \Delta_{0}}
    \right)_{\tau} = \frac{1}{\tau_{2}} \; \left[
    \eta(\tau) \overline{\eta(\tau)}\right]^{-2}\,.
    \label{eq:det0-mod}
\end{equation}
Note that there is a tradeoff between ``holomorphic factorization''
in \eqref{eq:det0} and ``modular invariance'' in \eqref{eq:det0-mod}.
Finally we observe that
\begin{equation}
    \left( \frac{\detp \partial_{1/2}}{N^{2}} \right)_{\tau'} =
    \bar{\tau}^{-1/2}\; 
    \left( \frac{\detp \partial_{1/2}}{N^{2}} \right)_{\tau}\,,
    \label{eq:mod-det-12}
\end{equation}
and
\begin{align*}
    \int_{M} \prod_{r=1}^{n} \frac{\mathfrak{z} \left(x_{r} z(\hb)
    \right)}{\vartheta\left(x_{r} z(\hb); \tau \right)} & \propto
    \int_{M} \prod_{r=1}^{n} \frac{\vartheta'(0;\tau) \left(x_{r}
    \tau_{2} \overline{\vartheta'(0;\tau)} \right)
    }{\vartheta\left(x_{r} \left(\tau_{2}
    \overline{\vartheta'(0;\tau)} \right); \tau \right)} \\
     & \propto \tau_{2}^{n} \, \vartheta'(0;\tau)^{n}
     \overline{\vartheta'(0;\tau)^{n}}\; \int_{M} \prod_{r=1}^{n} \frac{
     x_{r} }{\vartheta\left(x_{r} ; \tau \right)} \\
     & \propto \tau_{2}^{n} \, \eta(\tau)^{3n}
     \overline{\eta(\tau)^{3n}}\; \int_{M} \prod_{r=1}^{n} \frac{
     x_{r} }{\vartheta\left(x_{r} ; \tau \right)}\,.
\end{align*}
The modular transformation properties become
\begin{align*}
	 (\tau'_{2})^{n} \, \eta(\tau')^{3n}
	 \overline{\eta(\tau')^{3n}}\; \int_{M} \prod_{r=1}^{n} \frac{
	 x_{r} }{\vartheta\left(x_{r} ; \tau' \right)} &\propto
	 \left(\frac{\tau_{2}}{\tau \bar{\tau}}\right)^{n} \, 
	 (\tau\bar{\tau})^{3n/2}\eta(\tau)^{3n}
	 \overline{\eta(\tau)^{3n}}\; \int_{M} \prod_{r=1}^{n} \frac{
	 x_{r} }{\vartheta\left(x_{r} ; \tau' \right)} \\
      &\propto
	 (\tau\bar{\tau})^{n/2} \, \tau_{2}^{n} \, 
	 \eta(\tau)^{3n}\; 
	 \overline{\eta(\tau)^{3n}}\; \int_{M} \prod_{r=1}^{n} \frac{
	 x_{r} }{\vartheta\left((x_{r}\tau)/\tau ; -1/\tau \right)} \\
     &\propto
	(\tau\bar{\tau})^{n/2} \, \tau_{2}^{n} \, 
	\eta(\tau)^{3n}\; 
	\overline{\eta(\tau)^{3n}}\; \int_{M} \prod_{r=1}^{n} \frac{
	x_{r} }{\tau^{1/2}\vartheta\left(x_{r}\tau; \tau \right)} 
	\\
      &\propto (\tau\bar{\tau})^{n/2} \, \tau_{2}^{n} \,
      \eta(\tau)^{3n}\; \overline{\eta(\tau)^{3n}}\; \int_{M}
      \prod_{r=1}^{n} \frac{ x_{r}\tau
      }{\tau^{3/2}\vartheta\left(x_{r}\tau; \tau \right)} \\
      &\propto (\tau\bar{\tau})^{n/2} \, \tau_{2}^{n} \,
      \eta(\tau)^{3n}\; \overline{\eta(\tau)^{3n}}\;
      \frac{\tau^{n}}{\tau^{3n/2}} \int_{M} \prod_{r=1}^{n} \frac{
      x_{r} }{\vartheta\left(x_{r}; \tau \right)} \\
      &\propto \bar{\tau}^{n/2} \left( \tau_{2}^{n} \,
      \eta(\tau)^{3n}\; \overline{\eta(\tau)^{3n}}\;
      \int_{M} \prod_{r=1}^{n} \frac{
      x_{r} }{\vartheta\left(x_{r}; \tau \right)} \right),
\end{align*}
in agreement with the abstract result.  In fact
\begin{equation}
    \int_{M} \prod_{r=1}^{n} \frac{\mathfrak{z} \left(x_{r} z(\hb)
    \right)}{\vartheta\left(x_{r} z(\hb); \tau \right)} \propto
    \tau_{2}^{n} \; \overline{\vartheta'(0; \tau)^{n}}\;
    \int_{M} \prod_{r=1}^{n} \frac{ \vartheta'(0;\tau)
    }{\vartheta\left(x_{r}; \tau \right)/x_{r}}
    \label{eq:script-z}
\end{equation}
which is exactly \eqref{eq:last-nice} if we group a factor of
$\tau_{2}\, \overline{\vartheta'(0;\tau)}$ with each $x_{r}$.

The string genus is the product of three terms
\begin{equation*}
    \left(\frac{\vol \Sigma}{\det_{\onep} \Delta_{0}}
    \right)^{n} \times \left( \frac{\detp
    \partial_{1/2}}{N^{2}}\right)^{n} \times \int_{M}
    \prod_{r=1}^{n} \frac{\mathfrak{z} \left(x_{r} z(\hb)
    \right)}{\vartheta\left(x_{r} z(\hb); \tau \right)}\;.
\end{equation*}
Each term has different modularity and holomorphicity properties;
together they conspire to give something holomorphic which changes by
a phase under modular transformations.  There is no clear way of
subgrouping the terms pairwise; there is a tradeoff between
holomorphicity \emph{vs} modularity depending on the grouping.

Note that the string genus is 
\begin{equation}
    \frac{1}{\eta(\tau)^{2n}}\; \int_{M} \prod_{r=1}^{n} \frac{
    \vartheta'(0;\tau) }{\vartheta\left(x_{r}; \tau \right)/x_{r}}\,.
    \label{eq:string-genus-1}
\end{equation}
This should be compared with \eqref{eq:def-Zsc-1}. The fact that in 
$g=1$ the theta function is a function of one variable allowed us to 
scale out a variety of terms, namely all the non-holomorphic ones.

We can rewrite \eqref{eq:string-genus-1} as $\int_{M} 
\hat{s}(M,\tau)$ where $\hat{s}(M,\tau) = 
\hat{a}(M,\tau)/\eta(\tau)^{\dim M/2}$ and
\begin{equation}
    \hat{a}(M,\tau) = \prod_{j=1}^{\dim M/2} 
    \frac{ix_{j}/2\pi}{\sigma(ix_{j}/2\pi,\tau)}\;
    \label{eq:def-ahat}
\end{equation}
and $\sigma$ is the Weierstrass sigma function. $\hat{s}(M,\tau)$ is 
the string genus because $\int_{M} \hat{a}(M,\tau)$ is a modular form 
of weight $\dim M/2$, giving a homomorphism from string cobordism to 
the ring of modular forms.

The string genus is a generalization of Hirzebruch's genuses
\cite[Chapters 1, 2 \& 3]{Hirzebruch:book}  which are
homomorphisms from a cobordism theory ($\MSO^{*}$, $\MU^{*}$,
$\MSpin^{*}$) to the ring of integers using the appropriate power
series $Q(z)$ respectively given by ($L$, \text{Todd}, $\hat{A}$).
The homomorphism being $\int_{M} \prod_{j} Q(x_{j})$.

In the genus $1$ case, the analog of $Q(z)$ is the function of one
variable $(iz/2\pi)/\sigma(iz/2\pi,\tau)$.  When the genus $g>1$, see
Section~\ref{sec:geom-symp-det}, we get a genus from a subring of the
string cobordism ring, $\MString^{*}$, to a subring of the functions
on $\Teichspin(\Sigma)$.

\newpage

\section*{Nomenclature}

\begin{description}
    \item [{$\Met(\Sigma)$}] The space of metrics on
    $\Sigma$ with curvature $-1$

    \item [{$\Teich(\Sigma)$}] The Teichmuller space of
    $\Sigma$

  \item [{$\Teichodd(\Sigma)$}] The odd spin Teichmuller
  space of $\Sigma$

  \item [{$\qbasic$}] The basic quarf

  \item [{$\spin$}] A spin structure

  \item [{$\spinbasic$}] The spin structure associated to
  the basic quarf $\qbasic$
\end{description}


\section*{References}

\def\cprime{$'$}

\end{document}